\documentclass[journal,onecolumn]{IEEEtran}

\usepackage[utf8]{inputenc} 
\usepackage[T1]{fontenc}
\usepackage{url}
\usepackage{ifthen}
\usepackage{cite}
\usepackage[cmex10]{amsmath} 

\usepackage{algorithm,algpseudocode,tabularx,diagbox}
\usepackage{caption}
\usepackage{etoolbox}

\usepackage{circuitikz}
\usepackage{mathrsfs}

\usepackage{acronym,amsmath,amssymb,amsthm,dsfont,bbding,fancyhdr,graphicx,lastpage}

\usepackage{hyperref}

\newtheorem{definition}{Definition}[section]
\newtheorem{theorem}[definition]{Theorem}

\newtheorem{proposition}[definition]{Proposition}

\newcommand{\calA}{\mathcal{A}}

\newcommand{\bbC}{\mathbb{C}}

\newcommand{\calD}{\mathcal{D}}
\newcommand{\bbD}{\mathbb{D}}

\newcommand{\calE}{\mathcal{E}}
\newcommand{\bbE}{\mathbb{E}}

\newcommand{\bbI}{\mathbb{I}}

\newcommand{\calM}{\mathcal{M}}

\newcommand{\bbN}{\mathbb{N}}

\newcommand{\bbP}{\mathbb{P}}

\newcommand{\calQ}{\mathcal{Q}}

\newcommand{\bbR}{\mathbb{R}}

\newcommand{\calS}{\mathcal{S}}

\newcommand{\calU}{\mathcal{U}}

\newcommand{\calV}{\mathcal{V}}

\newcommand{\calX}{\mathcal{X}}

\newcommand{\calY}{\mathcal{Y}}

\newcommand{\argmax}{\mathop{\text{argmax}}}

\newcommand{\eqdef}{\triangleq}

\newcommand{\ket}[1]{{\left\vert{#1}\right\rangle}}

\renewcommand{\leq}{\leqslant} 
\renewcommand{\geq}{\geqslant} 

\begin{document}
\title{Active Hypothesis Testing for Quantum Detection of Phase-Shift Keying Coherent States} 

\author{%
  \IEEEauthorblockN{Yun-Feng Lo and Matthieu R. Bloch\\}
  \IEEEauthorblockA{School of Electrical and Computer Engineering,
                    Georgia Institute of Technology\\
                    Atlanta, Georgia, 30332-0250, United States
                    \\Email: ylo49@gatech.edu,~matthieu.bloch@ece.gatech.edu}
}

\maketitle

\begin{abstract}
   This paper explores the quantum detection of Phase-Shift Keying (PSK)-coded coherent states through the lens of active hypothesis testing, focusing on a Dolinar-like receiver with constraints on displacement amplitude and energy. With coherent state slicing, we formulate the problem as a controlled sensing task in which observation kernels have parameters shrinking with sample size. The constrained open-loop error exponent and a corresponding upper bound on the Bayesian error probability are proven. Surprisingly, the exponent-optimal open-loop policy for binary PSK with high dark counts is not simply time-sharing. This work serves as a first step towards obtaining analytical insights through the active hypothesis testing framework for designing resource-constrained quantum communication receivers.
\end{abstract}

\section{Introduction}
\label{sec:intro}

The problem of detecting and distinguishing quantum states is of fundamental importance in quantum information science. Quantum communication and sensing are ultimately limited by how well possibly nonorthogonal quantum states can be distinguished. The minimum value of the probability of error for distinguishing quantum states is colloquially known as the Helstrom limit~\cite{helstrom_detection_1967,helstrom_detection_1968,helstrom_quantum_1976} and the associated Positive Operater-Valued Measure (POVM) can be identified by solving a linear program~\cite{yuen_communication_1971,holevo_statistical_1972,yuen_optimum_1975}. However, experimentally implementing the optimal POVM remains a challenge, in general.

The special case of quantum optical states has attracted particular attention, and several receivers based on experimentally feasible optical components have been proposed. Direct detection receivers that measure the intensity of signals by counting photons are well studied and have been used, for instance, to characterize the Poisson channel capacity with coherent state inputs~\cite{shapiro_optical_1979,shapiro_ultimate_2005,martinez_spectral_2007}. More sophisticated receivers are required to distinguish coherent states with distinct phases. For binary phase shift keying (BPSK)-coded coherent states, Kennedy~\cite{kennedy_near-optimum_1973,shapiro_near-optimum_1980} proposed an architecture that displaces the incoming state before direct detection, by making the incoming state interfere with a strong local oscillator on a beam splitter. The optimal value of the displacement in Kennedy's receiver can be exactly characterized~\cite{takeoka_discrimination_2008,wittmann_near-optimal_2008,takeoka_near-optimal_2009}. Perhaps surprisingly, Dolinar~\cite{dolinar_optimum_1973, dolinar_class_1976} showed that adaptively controlling the displacement of the Kennedy receiver during detection significantly improves performance and even achieves the Helstrom limit. The Dolinar receiver can be interpreted as a decision-making problem in which, at each step, the decision maker attempts to maximize a utility in the form of an average mutual information~\cite{erkmen_dolinar_2011,chung_capacity_2017}. Unfortunately, generalizing the Dolinar receiver and achieving the Helstrom limit beyond the BPSK case has proved challenging. Not much is known about optimal architectures for distinguishing as few as three states~\cite{bondurant_near-quantum_1993,becerra_state_2011,becerra_m-ary-state_2011,becerra_photon_2015}.

The architecture of the Dolinar receiver exhibits striking similarities with active hypothesis testing~\cite{naghshvar_active_2010,naghshvar_sequentiality_2013,naghshvar_active_2013} and controlled sensing~\cite{nitinawarat_controlled_2013,nitinawarat_controlled_2013-1,nitinawarat_controlled_2013-2,nitinawarat_controlled_2015} problems in the statistics literature. In fact, one can view the Dolinar receiver as an instance of a POVM with classical parameters that can adapt~\cite{dimario_demonstration_2022,rodriguez-garcia_determination_2022,rodriguez-garcia_adaptive_2024}. Several works have already extended the classical results of active hypothesis testing to the quantum setting~\cite{li_optimal_2021,martinez_vargas_quantum_2021,fields_sequential_2024}. 
Nevertheless, these ideas have thus far not offered many insights into new experimental architectures that would help realize efficient quantum detection in a laboratory. Moreover, the unavoidable experimental imperfections occurring in the real world break any optimality claim of the Dolinar receiver, and little is known about the performance of non-ideal detectors~\cite{yuan_kennedy_2020,yuan_optimally_2021}.

Motivated by a recent experimental result~\cite{Cui2022Quantum} on quantum receiver enhanced by adaptive learning, effectively combining the Dolinar receiver with reinforcement learning, we investigate the problem of quantum detection of PSK-coded coherent states through the lens of active hypothesis testing. Specifically, we consider a situation in which a Dolinar-like receiver attempts to discriminate quantum states using sequential decision-making policy subject to constraints on the peak and average squared displacement that can be applied. Such constraints are motivated by applications to deep space communications using resource-constrained satellites. We also consider the case in which the noises and imperfections are collectively modeled as dark counts. This, arguably extremely simplified situation, serves as an exemplar to leverage the breadth of knowledge in sequential decision-making to develop analytical insight into the quantum detection problem relevant for experimental systems. The main contributions of the paper are 1) Formulation of the single-shot PSK-coded coherent state discrimination problem as an active hypothesis testing / controlled sensing problem; 2) Proof that even in the unconventional setting of observation kernels having parameters shrinking with sample size, a properly defined constrained open-loop exponent takes a form similar to that in classical active hypothesis testing / controlled sensing problems; 3) Proof that when dark counts is high in the BPSK case, the exponent-optimal constrained open-loop control policy is not time-sharing between the Kennedy displacement and zero.
Our numerical results suggest the potential of the approach to develop competitive resource-constrained receivers.

The remainder of the paper is organized as follows. We introduce notations in Section~\ref{sec:notations} and formally introduce the system model in Section~\ref{sec:problem}. We present our main results in Section~\ref{sec:results}, as well as numerical results illustrating the benefits of an active hypothesis testing approach to quantum detection. We relegate proof details to the Appendix~\ref{appen-proof-thm} and~\ref{appen-proof-prop} to streamline the presentation. 

\section{Notations}
\label{sec:notations}
Throughout this paper, $\log(\cdot)$ denotes natural logarithm. $\mathbb{I}\{\cdot\}$ denotes the indicator function. $\text{Poi}(\lambda)$ denotes the Poisson distribution with rate $\lambda$.
Let $\bbN\eqdef\{1,\cdots\}$ and $\bbN_0\eqdef\{0,1,\cdots\}.$ Let $\bbR$ be the set of real numbers, $\bbR_{>0}\eqdef\{x\in\bbR:x>0\}$ and $\bbR_{\geq 0}\eqdef\{x\in\bbR:x\geq 0\}$. Let $\bbC$ be the set of complex numbers, and $i$ be the imaginary unit.
For $n\in\bbN$, we denote $[n]\eqdef\{1,\cdots,n\}$.
For $n\in\bbN$, a sequence of random variables $(X_1,\cdots,X_n)$ is denoted $X^n$; the corresponding sequence of realizations $(x_1,\cdots,x_n)$ is denoted $x^n$. Also, $X^0 \eqdef \emptyset \eqdef x^0$, where $\emptyset$ denotes the empty symbol.  
For any set $\calX$, let $|\calX|$ denote its cardinality. For any $z\in\bbC$, let $|z|$ denote its modulus. For any $R>0$, let $\calD(R)\eqdef\{z\in\bbC:|z|\leq R\}$. For any $z\in\bbC$ and set $\calX\subset\bbC$, let $z\calX\triangleq\{zx:x\in\calX\}$.
For any set $\calX$ and $N\in\bbN$, let $\calX^N$ denote the $N$-fold Cartesian product $\calX$. For any set $\calX$, let $\mathscr{F}_{\calX}$ denote a sigma algebra on $\calX$, and let $2^\calX$ denote the power set of $\calX$ in case $|\calX|<\infty$. For a set $\calX \subset \bbC,$ let $\mathscr{B}_{\calX}$ denote the Borel sigma-algebra on $\calX$. Let $\mu_{\mathbb{N}_0}$ denotes the counting measure on $\mathbb{N}_0$.

The Chernoff $s$-divergence between densities $p$ and $q$ with respect to (w.r.t.) a common dominating measure $\mu$ on the space $\Omega$ is defined as
    $
        \mathbb{C}_s(p \| q; \mu)
        \triangleq
        -\log \int_\Omega p^s q^{1-s} \mathrm{d}\mu
    $
for $s\in[0,1]$. Let $\mathbb{D}(p \| q; \mu)$ be the relative entropy (or Kullback-Liebler divergence) between densities $p$ and $q$ w.r.t. $\mu$, i.e.,
$
    \mathbb{D}(p \| q; \mu)
    \triangleq
    \int_\Omega q \log \frac{p}{q} \mathrm{d}\mu
    .
$

\section{Problem setting}
\label{sec:problem}

We consider the model illustrated in Fig.~\ref{fig:model}, in which a transmitter sends a coherent state signal $S$ from the set $\calS=\{ \ket{\alpha \textnormal{e}^{i\phi_m}} \}_{m\in\calM}$ for some known $\alpha>0$ and $\phi_m\in[0,2\pi)$ for each $m\in\calM\eqdef\{0,1,\ldots,|\calM|-1\}.$ For simplicity, we model imperfections and noise collectively as a dark count rate $\lambda_\text{d}$, while losses can be factored into the definition of states in $\calS$~\cite{chung_capacity_2017}. The receiver is modeled as a resource-constrained Kennedy-Dolinar architecture that can cause the incoming state to interfere with a local reference signal $
$ subject to amplitude and average energy constraints, before the mixed signal is fed into a photon-number-resolving detector (PNRD). Specifically, a highly transmissive (with transmissivity $\gamma\approx 1$) beam splitter is used to create a displacement operation on the input signal $S$ and produce a mixed signal $S+u$, where the amount of displacement $u$ is controlled by the local reference signal $\ell$ via the relationship $\ell=u/\sqrt{1-\gamma}$. The objective of the receiver is to design a detection policy in the form of a sequence of displacements together with a decision rule to identify the received state. Concretely, we restrict this paper to the scenario of Bayesian minimum error discrimination (as opposed to unambiguous state discrimination~\cite{becerra_implementation_2013,zhuang_ultimate_2020,sidhu_linear_2023}), i.e., the goal of the detection policy is to minimize the Bayesian error probability of detecting the given coherent state. For simplicity, we assume a uniform prior. Moreover, we consider a ``one-shot discrimination" setting in which only a single copy of the coherent states is transmitted. Because of imperfections and noise, the Dolinar receiver \cite{dolinar_optimum_1973} is no longer known to be optimal; the optimized displacement receiver (ODR) first proposed by Takeoka \& Sasaki \cite{takeoka_discrimination_2008,wittmann_near-optimal_2008,takeoka_near-optimal_2009} incorporates dark count, but generally requires a displacement larger than that of the near-optimal Kennedy receiver~\cite{kennedy_near-optimum_1973}. 

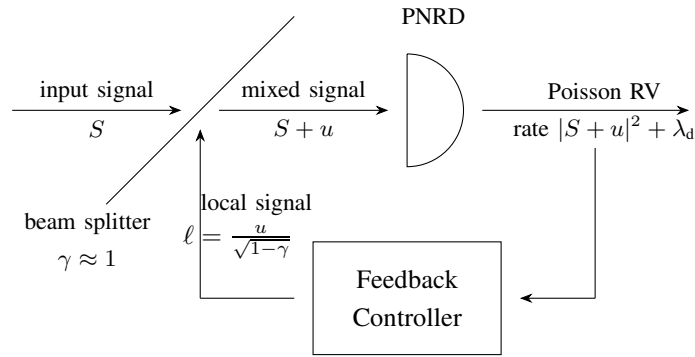
\begin{figure}[!ht]
  \centering
    \begin{circuitikz}
    \draw [->, >=Stealth] (0,8.75) -- (2.25,8.75);
    \node [font=\small] at (1.125,9) {input signal};
    \node [font=\small] at (1.125,8.5) {$S$};
    \draw [short] (3.75,10) -- (1.25,7.5);
    \node [font=\small] at (1,7.25) {beam splitter};
    \node [font=\small] at (1,6.75) {$\gamma\approx 1$};
    \draw [->, >=Stealth] (2.75,8.75) -- (5,8.75);
    \node [font=\small] at (3.875,9) {mixed signal};
    \node [font=\small] at (3.875,8.5) {$S+u$};
    \draw  (5.25,9.5) -- (5.25,8) arc[start angle=270, delta angle=180, radius=0.75cm] -- (5.25,9.5) ;
    \node [font=\small] at (5.6,10) {PNRD};
    \draw [->, >=Stealth] (6.25,8.75) -- (9,8.75);
    \node [font=\small] at (7.875,9) {Poisson RV};
    \node [font=\small] at (7.875,8.5) {rate $ |S+u|^2 + \lambda_\textnormal{d} $};
    \draw [short] (7.75,8.25) -- (7.75,6.25);
    \draw [->, >=Stealth] (7.75,6.25) -- (6.75,6.25);
    \draw  (4,7) rectangle (6.5,5.5);
    \node [font=\normalsize] at (5.25,6.5) {Feedback};
    \node [font=\normalsize] at (5.25,6) {Controller};
    \draw [short] (2.5,6.25) -- (3.75,6.25);
    \draw [->, >=Stealth] (2.5,6.25) -- (2.5,8.5);
-    \node [font=\small] at (3.25,7.5) {local signal};
    \node [font=\normalsize] at (3,7) {$\ell=\frac{u}{\sqrt{1-\gamma}}$
    };
    \end{circuitikz}
 
  \caption{Sequential detection of optical quantum states. Adapted from \cite{chung_capacity_2017}.}
  \label{fig:model}
\end{figure}

Slicing~\cite{assalini_revisiting_2011,da_silva_achieving_2013,nair_realizable_2014} is known to improve the error rate performance of coherent state discrimination, by translating the one-shot discrimination task into a multi-copy discrimination problem. For example, the quadrature phase-shift keying receiver (QPSK) of Bondurant \cite{bondurant_near-quantum_1993} slices coherent state signals w.r.t. time \cite{nair_realizable_2014}, while the sequential waveform nulling receiver of Nair et al. \cite{nair_realizable_2014} slices coherent state signals w.r.t. amplitude. In this work, we only consider the slicing of flat-top temporal pulses w.r.t. time, as illustrated in Fig.~\ref{fig:slicing-input}. In particular, we assume that the input signals from the set $\calS$ of coherent states $\{ \ket{\alpha \textnormal{e}^{i\phi_m}} \}_{m\in\calM}$ to be discriminated are temporal modes with absolute amplitude $A$ (with $A>0$) and duration $T>0$, where $\alpha=\sqrt{n_\text{s}}$, and $n_\text{s}=A^2 T$ is the mean photon number corresponding to these coherent state signals. Denote $N$ as the number of time slices (so $N\in\mathbb{N}$) and $\Delta$ as the duration of each slice so that $T=N\Delta$. Since, without loss of generality, we can set $T=1$ \cite{zoratti_agnostic_2021} (while interpreting the dark count rate $\lambda_\textnormal{d}$ as the probability of dark count happening in the duration $T$), in the following, we identify $\alpha=A$ and $n_\text{s}=A^2$. We also define the signal-to-noise ratio (SNR) $R_\text{SN}\triangleq \alpha^2/\lambda_\text{d}\in(0,\infty)$.

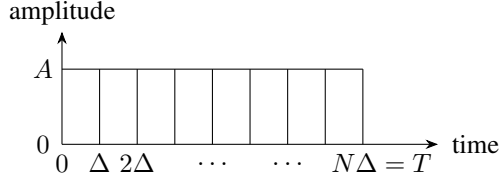
\begin{figure}[!ht]
  \centering
    \begin{tikzpicture}
        \draw [->, >=Stealth] (0,0) -- (5,0);
        \node [font=\normalsize] at (5.5,0) {time};
        \node [font=\normalsize] at (-0.25,0) {$0$};
        \draw [->, >=Stealth] (0,0) -- (0,1.5);
        \node [font=\normalsize] at (0,1.75) {amplitude};
        \node [font=\normalsize] at (0,-0.25) {$0$};
        \draw [-] (0,1) -- (4,1);
        \node [font=\normalsize] at (-0.25,1) {$A$};
        \draw [-] (0.5,0) -- (0.5,1);
        \node [font=\normalsize] at (0.5,-0.25) {$\Delta$};
        \draw [-] (1,0) -- (1,1);
        \node [font=\normalsize] at (1,-0.25) {$2\Delta$};
        \draw [-] (1.5,0) -- (1.5,1);
        \draw [-] (2,0) -- (2,1);
        \node [font=\normalsize] at (2,-0.25) {$\ldots$};
        \draw [-] (2.5,0) -- (2.5,1);
        \draw [-] (3,0) -- (3,1);
        \node [font=\normalsize] at (3,-0.25) {$\ldots$};
        \draw [-] (3.5,0) -- (3.5,1);
        \draw [-] (4,0) -- (4,1);
        \node [font=\normalsize] at (4.25,-0.25) {$N\Delta=T$};

    \end{tikzpicture}
    
  \caption{Temporal slicing of the input coherent state pulse with total duration $T$ and constant amplitude $A$ into $N$ slices, each of duration $\Delta = T / N$.}
  \label{fig:slicing-input}
\end{figure}

We assume that the local reference signal $\ell(t)=u(t)/\sqrt{1-\gamma}$ in Fig.~\ref{fig:model} stays constant for the duration of each slice. Hence, the same condition holds for the displacement $u(t)$, as depicted in Fig.~\ref{fig:slicing-control} for a general displacement signal. In Fig.~\ref{fig:slicing-control}, only the amplitude of $u(t)$ is plotted but, in general, the phase of $u(t)$ can also change between slices. First, we observe that since the transmissivity $\gamma$ is close to $1$, the local reference signal $\ell(t)$ is strong even if $|u(t)|$ is small. Moreover, it is known that in the BPSK case, when the priors are equal, the Dolinar receiver requires an infinite amount of displacement at the start (i.e., at $t=0$, $|u(t)|=\infty$~\cite{geremia_distinguishing_2004}); it can also be shown via a direct calculation that because of this singularity at $t=0$, the energy $\mathcal{E} \triangleq \int_0^T |u(t)|^2 \mathrm{d}t$ (up to some proportionality constant) required to generate the displacement signal $u(t)$ for the Dolinar receiver is infinite. Even when the priors are not exactly equal but close to being equal (which is a natural assumption when we do not know if one coherent state is more probable than the other), the Dolinar receiver requires several times the energy of the Kennedy receiver or the ODR. Motivated by these observations, we focus on exploring the other end of the control energy spectrum, i.e., when the controller can only use a control signal $u(t)$ that has an amplitude at most $\alpha=A$, and, simultaneously, has energy (or average squared amplitude) at most $n_\text{s}=\alpha^2$. 

\begin{figure}[!ht]
  \centering
    
    \begin{tikzpicture}
        \draw [->, >=Stealth] (0,0) -- (5,0);
        \node [font=\normalsize] at (5.5,0) {time};
        \draw [->, >=Stealth] (0,0) -- (0,2);
        \node [font=\normalsize] at (0,2.125) {amplitude};
        
        \node [font=\normalsize] at (-0.5,1.75) {$u(t)$};
        \draw [-,thick] (0,1.5) -- (0.5,1.5);
        \draw [-,thick] (0.5,1.5) -- (0.5,1);
        \draw [-,thick] (0.5,1) -- (1,1);
        \draw [-,thick] (1,1) -- (1,0.5);
        \draw [-,thick] (1,0.5) -- (1.5,0.5);
        \draw [-,thick] (1.5,0.5) -- (2,0.5);
        \draw [-,thick] (2,0.5) -- (2,1);
        \draw [-,thick] (2,1) -- (2.5,1);
        \draw [-,thick] (2.5,1) -- (2.5,1.5) ;
        \draw [-,thick] (2.5,1.5) -- (3,1.5);
        \draw [-,thick] (3,1.5) -- (3,1.5) ;
        \draw [-,thick] (3,1.5) -- (3.5,1.5);
        \draw [-,thick] (3.5,1.5) -- (3.5,0.5);
        \draw [-,thick] (3.5,0.5) -- (4,0.5);
        \draw [-,thick] (4,0.5) -- (4,0);
        \draw [-,dashed] (0.5,0) -- (0.5,1.75);
        \draw [-,dashed] (1.0,0) -- (1.0,1.75);
        \draw [-,dashed] (1.5,0) -- (1.5,1.75);
        \draw [-,dashed] (2.0,0) -- (2.0,1.75);
        \draw [-,dashed] (2.5,0) -- (2.5,1.75);
        \draw [-,dashed] (3.0,0) -- (3.0,1.75);
        \draw [-,dashed] (3.5,0) -- (3.5,1.75);
        \draw [-,dashed] (4.0,0) -- (4.0,1.75);
        \draw [-,dashed] (0,1.5) -- (4.5,1.5);
        \draw [-,dashed] (0,1) -- (4.5,1);
        \draw [-,dashed] (0,0.5) -- (4.5,0.5);
        \draw [-,dashed] (0,0) -- (4.5,0);
    \end{tikzpicture}
        
  \caption{A control signal of duration $T$ and generally non-constant amplitudes across $N$ intervals of duration $\Delta$.}
  \label{fig:slicing-control}
\end{figure}
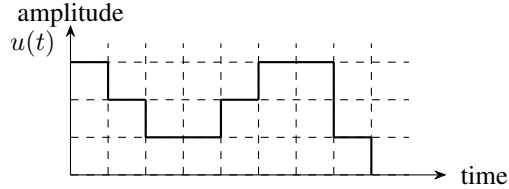

Our first contribution is the observation that, if formulated carefully, the problem of designing a detection policy for this resource-constrained Kennedy-Dolinar type receiver can be considered as an active hypothesis testing \cite{naghshvar_active_2010,naghshvar_sequentiality_2013,naghshvar_active_2013} or a controlled sensing \cite{nitinawarat_controlled_2013,nitinawarat_controlled_2013-1,nitinawarat_controlled_2013-2} problem. Indeed, the number of slices $N$ can be considered the sample size; the hypothesis set is $\mathcal{M}$ and $M\in\mathcal{M}$ is the true hypothesis considered as a random variable (r.v.); the observation alphabet $\mathcal{Y}=\mathbb{N}_0$ captures the output possibilities of the PNRD
; and, because of the limited precision of the feedback controller, the control action set $\mathcal{U}$ is chosen to be a suitably quantized version of the closed disk $\calD(\alpha R_\textnormal{CA})$, where $R_\text{CA}\in (0,\infty)$ is the control peak amplitude to coherent state amplitude ratio. Moreover, we assume that $0\in\mathcal{U}$. Note that $\calU\subset\calD(\alpha R_\textnormal{CA})$ already encodes an amplitude constraint on the feedback control signal $u(t)=\sum_{n=1}^N u_n \mathbb{I}\{ t \in [(n-1)\Delta,n\Delta)\}$ by restricting that the control actions $u_n\in\mathcal{U}$ for $n\in[N]$. Because of the independence of slices \cite{assalini_revisiting_2011,da_silva_achieving_2013,nair_realizable_2014} resulting from coherent state slicing, the photon counter observation $Y_n$ (a r.v.) during the time interval $[(n-1)\Delta,n\Delta)$ satisfies the ``stationary Markovity assumption" of \cite{nitinawarat_controlled_2013}, i.e., given the current control $U_n=u_n$ (where $U_n$ denotes the $n$th generally random control action, as control policies are generally stochastic, and $u_n$ can be considered as the realization of $U_n$), the current observation $Y_n$ is conditionally independent of past actions $U^{n-1}$ and past observations $Y^{n-1}$. Following the stationary Markovity assumption, we can define a set of observation kernels (which are probability distributions) $\{p_m^u\}_{m\in\mathcal{M}}^{u\in\mathcal{U}}$ that characterize the active hypothesis testing/controlled sensing problem: given the past observations $Y^{n-1}=y^{n-1}$, the past actions $U^{n-1}=u^{n-1}$, the current action $U_n=u_n$, and the hypothesis $M=m$, the current observation $Y_n$ follows the probability distribution $p_m^{u_n}$. A technical assumption commonly assumed in controlled sensing (see, e.g., \cite{nitinawarat_controlled_2013}) is that for each action $u\in\mathcal{U}$, the probability distributions $\{p^u_m\}_{m\in\mathcal{M}}$ are all absolutely continuous w.r.t. a dominating measure $\mu_u$ (which, in general, depends on the action $u$). Then, we can equivalently think of each $p^u_m$ as its Radon-Nikodym derivative (``density") w.r.t. $\mu_u$. For the coherent state discrimination problem, we can set $\mu_u=\mathbb{N}_0$ for all actions $u\in\mathcal{U}$. Then, following the Poissonian statistics of photon detection, for each $u\in\mathcal{U}$ and $m\in\calM$, $p_m^{u}(y)$ is the probability mass function (PMF) of a Poisson distribution with rate $\lambda^u_m(\alpha,\lambda_\text{d},\Delta) \eqdef \left[|\alpha\text{e}^{i\phi_m}+u|^2+\lambda_\text{d}\right]\Delta$. We also write $p_m^u=p_m^u[\alpha,\lambda_\text{d},\Delta]$ to highlight its dependence on parameters. It is worth noting that our problem setting does not exactly fit that of active hypothesis testing/controlled sensing in, e.g., \cite{naghshvar_active_2010,naghshvar_sequentiality_2013,naghshvar_active_2013,nitinawarat_controlled_2013,nitinawarat_controlled_2013-1,nitinawarat_controlled_2013-2}, because the observation kernels $\{p_m^u\}^{u\in\mathcal{U}}_{m\in\mathcal{M}}$ depend on the number of observations $N=1/\Delta$ due to coherent state slicing.

In the terminology of active hypothesis testing/controlled sensing, the detection policy we seek is in general a non-sequential/fixed-sample-size and adaptive/causal policy. The policy is non-sequential because after choosing a small enough but fixed $\Delta$, we obtain a fixed number of durations $N$ (and hence the same number of observation samples). The policy is adaptive because the control actions $U_n$ on the $n$th $\Delta$-duration, in general, depend on all the past observations $U^{n-1}$ (recall the ``feedback controller" structure in Fig.~\ref{fig:model}). We also denote such a fixed-sample-size causal policy $q$ as a collection $\{ q_n( u_n | y^{n-1} , u^{n-1}) \}_{n\in[N]} $ of ``control kernels," i.e., each $q_n( u_n | y^{n-1} , u^{n-1})$ is the probability mass function on $\mathcal{U}$ of the $n$th control action $U_n$ given the past observations $Y^{n-1}=y^{n-1}$ and the past controls $U^{n-1}=u^{n-1}$. As a subset of the causal control policies, open-loop control policies are those satisfying the condition
\begin{align*}
q_n( u_n | y^{n-1} , u^{n-1})=q_n( u_n | u^{n-1}), \forall n\in[N],
\end{align*}
namely, the current action can only depend on past actions, but not past observations.

To precisely define the average energy constraint as well as the Bayesian error probability of a detection policy, we need a few definitions. First, the joint observation-control measure $\mathbb{P}_m$ conditioned on the hypothesis $M=m$ is defined via the joint density 
\begin{align*}
    p_m(y^N,u^N) 
    \triangleq 
    \prod_{n\in[N]} 
    \left[
    q_n( u_n | y^{n-1} , u^{n-1}) p_m^{u_n}(y_n)
    \right]
\end{align*} 
in the sense that for every joint event $\mathcal{A}=(\mathcal{A}_1,\mathcal{A}_2) \in \mathscr{F}_{\mathcal{Y}^N} \times 2^{(\mathcal{U}^N)}$ we have
\begin{align}
    \begin{split}
    \label{eq:def-joint-obser-contr-meas}
    &\mathbb{P}_m
    \left\lbrace
    (Y^N,U^N) \in \mathcal{A}
    \right\rbrace
    \\\triangleq~&
    \sum_{ u^N \in \mathcal{A}_2 }
    \int_{ y^N \in \mathcal{A}_1 }
    p_m ( y^N , u^N )
    \prod_{n\in [N]} \mathrm{d}\mu_{u_n}( y_n )
    .
    \end{split}
\end{align}
The expectation $\mathbb{E}_m[\cdot]$ is defined accordingly. The average energy constraint thus reads
\begin{align}
    \label{eq:avg_energy_constr}
    \mathbb{E}_m
    \left[
    \frac{1}{N} \sum_{n\in[N]} E(u_n) 
    \right]
    \leq 
    \mathcal{E}
    ,
    \forall m\in\mathcal{M},
\end{align}
where $E(\cdot):\mathcal{U}\to\mathbb{R}_{\geq 0}$ is the energy function of an action, and $\mathcal{E}\geq 0$ is the allotted control energy. We further define the control energy to coherent state energy ratio $R_\text{CE}\in(0,\infty),$ namely, $R_\text{CE}\eqdef\calE/\alpha^2$. This energy constraint implicitly restricts the admissible control policies since each policy induces a corresponding expectation $\mathbb{E}_m[\cdot]$. For the coherent state discrimination task, we choose $E(u)\triangleq |u|^2$. We also respectively denote $\mathcal{Q}_\text{CC}(\mathcal{U}^N,\mathcal{E})$, $\mathcal{Q}_\text{OL}(\mathcal{U}^N,\mathcal{E})$ as the set of length-$N$ causal-control, open-loop policies on the action set $\mathcal{U}$ satisfying the average energy constraint \eqref{eq:avg_energy_constr}. As mentioned earlier, we will focus on the regime $\mathcal{E}\in [0,\alpha^2]$, i.e., $R_\textnormal{CE}\leq 1$.

A detection policy (also called a test) $\tau=(q,\delta)$ consists of a control policy $q$ and a decision rule $\delta:\mathcal{Y}^N\times\mathcal{U}^N\to\mathcal{M}$ that is measurable w.r.t. $\bbP_m$ for each $m\in\calM$. The Bayesian error probability of a detection policy $\tau$ over the observation kernels $\{p_m^u\}^{u\in\mathcal{U}}_{m\in\mathcal{M}}$ is defined as
\begin{align*}
    P_\text{e}
    \left(
    \{p_m^u\}^{u\in\mathcal{U}}_{m\in\mathcal{M}}
    ,\tau
    \right)
    \triangleq
    \sum_{m\in\mathcal{M}}
    \pi_m 
    \mathbb{P}_m
    \left[ 
    \delta\left(Y^N,U^N\right) \neq m 
    \right]
\end{align*}
where $\pi_m$ is the prior probability of hypothesis $m\in\mathcal{M}$. We assume a uniform prior, i.e., $\pi_m=1/|\calM|$ for all $m\in\calM$. The set of optimal detection policies minimizing the Bayesian error probability is
$\arg\min_{\tau}  P_\text{e}\left(\{p_m^u\}^{u\in\mathcal{U}}_{m\in\mathcal{M}},\tau\right),$
which is generally difficult to characterize in the non-sequential adaptive (i.e., fixed-sample-size causal) setting \cite{naghshvar_active_2013, nitinawarat_controlled_2013}. Even in the case of binary coherent state discrimination, in which the Bayesian error probability takes the simpler form 
\begin{align*}
    &
    P_\text{e}
    \left(
    \{p_m^u\}^{u\in\mathcal{U}}_{m\in\{0,1\}}
    ,\tau
    \right)
    \\=~&
    \frac{1}{2}
    \mathbb{P}_0
    \left[ 
    \delta\left(Y^N,U^N\right) =1 
    \right]
    +
    \frac{1}{2}
    \mathbb{P}_1
    \left[ 
    \delta\left(Y^N,U^N\right) =0 
    \right]
    ,
\end{align*}
the set of optimal detection policies remains elusive except in ideal settings such as the noiseless case without amplitude and average energy constraint, where one of the known optimal detection policies that achieves the Helstrom bound is the Dolinar receiver \cite{dolinar_optimum_1973}.

A coarser concept for optimality in detection policies is that of exponent optimality (for the definition of such error exponents, see \cite{naghshvar_active_2013, nitinawarat_controlled_2013}). However, in the non-sequential adaptive (i.e., fixed-sample-size causal) setting, even the characterization of the error exponent remains an open problem; only upper and lower bounds are known \cite{naghshvar_active_2013, nitinawarat_controlled_2013}, although their tightness is unclear. In light of this difficulty, we focus on characterizing detection policies with \emph{open-loop} control that not only are exponent-optimal but also provide a corresponding exponentially-decaying upper bound on the minimum Bayesian error probability. Because the observation kernels $\{p_m^u\}^{u\in\mathcal{U}}_{m\in\mathcal{M}}$ depend on the number of observations $N$, we should be careful in defining the error exponent, not w.r.t. the number of slices $N$ but w.r.t. the ``energy" or mean photon number $n_\text{s}$ of the input coherent states, in the same spirit as the \emph{error probability exponent (EPE)} definition of coherent state receivers in \cite{nair_realizable_2014}. 
In particular, to carefully define the optimal amplitude and average energy-constrained open-loop error exponent w.r.t. $n_\text{s}=\alpha^2$, we fix the operating parameters $R_\text{SN}$, $R_\text{CA}$ and $R_\text{CE}$. We also specify a $K\in\bbN$ that determines the fineness of control set discretization. Then, for each $\alpha>0$, we associate the allotted average energy $\calE_\alpha\triangleq\alpha^2 R_\text{CE}$ and the set of controls $\calU_{\alpha,K} \triangleq \alpha\calV_K$ where $\calV_K \eqdef R_\text{CA}\{ \rho \textnormal{e}^{i\theta} ~:~ \rho \in \{0,1/K,\ldots,1\}, \theta \in \{0,2\pi/K,\ldots,2\pi(1-1/K)\}\}$. Then the optimal amplitude and average energy constrained open-loop error exponent w.r.t. $n_\text{s}=\alpha^2$ is formally defined as 

\begin{align*}
    &\beta_\textnormal{OL}
    \triangleq
    \beta_\textnormal{OL}
    ( R_\textnormal{SN} , R_\textnormal{CA} , R_\textnormal{CE} )
    \\
    \triangleq~&
    -
    \liminf_{ K , N \to \infty }
    \liminf_{ \alpha \to \infty }
    \inf_{ q \in \mathcal{Q}_\text{OL}( \calU_{\alpha,K}^N, \calE_\alpha ) }
    \inf_{ \delta : \mathcal{Y}^N \times \calU_{\alpha,K}^N \to \mathcal{M} }
    \\&
    \frac{1}{ \alpha^2 }
    \log
    P_\text{e}
    \left(
    \left\lbrace p_m^u \left[ \alpha , \frac{\alpha^2}{R_\textnormal{SN}} , \frac{1}{N} \right]
    \right\rbrace^{u\in\calU_{\alpha,K}}_{m\in\mathcal{M}}
    , ( q , \delta )
    \right)
    .
\end{align*}

\section{Main results}
\label{sec:results}

    \subsection{Theoretical results}
    \label{subsec:theory}
    
    \begin{theorem} 
    \label{thm:char}
    Given $R_\textnormal{SN}$, $R_\textnormal{CA}$ and $R_\textnormal{CE}$, the optimal amplitude and average energy constrained open-loop error exponent w.r.t. $n_\textnormal{s}=\alpha^2$ can be characterized as
     \begin{align*}
        &
        \beta_\textnormal{OL}(R_\textnormal{SN}, R_\textnormal{CA}, R_\textnormal{CE})
        \\
        =~&
        \sup_{
        Q \in \mathcal{Q}(R_\textnormal{CA},R_\textnormal{CE})
        }
        \min_{\substack{(\ell,m)\in\calM^2\\\ell<m}}
        \max_{s\in[0,1]}
        \bbE_{V\sim Q}
        \left[
        \mathbb{C}_s
        \left( P_\ell^V \| P_m^V ; \mu_{\bbN_0} 
        \right)
        \right]
        ,
    \end{align*}   
    where for each $v\in\bbC$ and $m\in\calM$, we define the distribution
    $P_m^v \triangleq \textnormal{Poi}(\Lambda_m^v)$ with the rate  
    $\Lambda_m^v \eqdef \Lambda_m(v) \eqdef | v + \textnormal{e}^{i\phi_m} |^2 + r_\textnormal{SN}$, where $r_\textnormal{SN}\eqdef 1/R_\textnormal{SN}$.  The set $\calQ(R_\textnormal{CA},R_\textnormal{CE})$ contains all probability distributions supported on $\calD(R_\textnormal{CA})$ with a second-moment constraint $R_\textnormal{CE}$, i.e.,
        \begin{align*}
            &\calQ(R_\textnormal{CA},R_\textnormal{CE})
            \triangleq
            \bigg\{
            Q: \mathscr{B}_{ \calD(R_\textnormal{CA}) } \to [0,1]
            ~\bigg\vert~
            \\ & \qquad \qquad \qquad \qquad
            Q( \calD(R_\textnormal{CA}) )
            =
            1
            ,
            \bbE_{V\sim Q}
            \left[
            |V|^2
            \right]
            \leq 
            R_\textnormal{CE}
            \bigg\}
            .
        \end{align*}
    Moreover, the following upper bound holds for the minimal amplitude and average energy constrained open-loop Bayesian error probability
        \begin{align*}
            &
            \liminf_{ K,N \to \infty }
            \inf_{ q \in \mathcal{Q}_\text{OL}( \calU_{\alpha,K}^N, \calE_\alpha ) }
            \inf_{ \delta : \mathcal{Y}^N \times \calU_{\alpha,K}^N \to \mathcal{M} }
            \\
            &P_\textnormal{e}
            \left(
            \left\lbrace p_m^u \left[ \alpha , \frac{\alpha^2}{R_\textnormal{SN}} , \frac{1}{N} \right]
            \right\rbrace^{u\in\calU_{\alpha,K}}_{m\in\mathcal{M}}
            , ( q , \delta )
            \right)
            \\
            &\leq 
            (|\calM|-1)
            \exp( 
            - \alpha^2
            \beta_\textnormal{OL}(R_\textnormal{SN}, R_\textnormal{CA}, R_\textnormal{CE})
            )
            .
        \end{align*}
     \end{theorem}
     \begin{proof}
         See Appendix~\ref{appen-proof-thm}.
     \end{proof}
     
     \begin{proposition}
     \label{prop:policy}
     Set $\calM=\{0,1\}$ with $H_0:~\ket{\alpha_0}=\ket{-\alpha}$ and $H_1:~\ket{\alpha_1}=\ket{\alpha}$ for some $\alpha>0$.
     Let $R_\textnormal{CA}=1$ and $R_\textnormal{CE}\leq 1$.
     When the SNR $R_\textnormal{SN}$ is large enough, there is an exponent-optimal open-loop control policy that is a time-sharing policy between the zero displacement $0$ and the Kennedy displacement $\alpha$. However, there exists some pair of SNR $R_\textnormal{SN}$ and $R_\textnormal{CE}$ values under which any open-loop control policy that is time-sharing between $0$ and $\alpha$ is not exponent-optimal.
     \end{proposition}
    \begin{proof}
         See Appendix~\ref{appen-proof-prop}.
     \end{proof}

    \subsection{Simulation results}
    \label{subsec:simulation}
        We present simulation results for the case $\calM=\{0,1\}$ with $H_0:~\ket{\alpha_0}=\ket{-\alpha}$ and $H_1:~\ket{\alpha_1}=\ket{\alpha}$ for some $\alpha>0$. In Fig.~\ref{fig:Pe_vs_n_s_high_SNR} and Fig.~\ref{fig:Pe_vs_n_s_low_SNR}, we plot the probability of error versus mean photon number $n_\textnormal{s}=\alpha^2$ of our method (``Ours"), the homodyne detector and the Helstrom bound. In our method, the error probability upper bound in Theorem~\ref{thm:char} is plotted with an improved pre-factor of $\frac{1}{2}$ instead of $|\calM|-1=1$, as can be shown to be valid in the binary case. We choose $R_\textnormal{CA}=1$, $R_\textnormal{CE}=1$ for both figures; the SNR $R_\textnormal{SN}=10^6$ in Fig.~\ref{fig:Pe_vs_n_s_high_SNR} and $R_\textnormal{SN}=10^2$ in Fig.~\ref{fig:Pe_vs_n_s_low_SNR}. The theoretical probability of error without considering dark count are plotted for the homodyne detector and the Helstrom bound. We observe that our method outperforms the homodyne detector for $\alpha^2\gtrsim 1$ in ``high" SNR (Fig.~\ref{fig:Pe_vs_n_s_high_SNR}) and for $\alpha^2\gtrsim 1.8$ in ``low" SNR (Fig.~\ref{fig:Pe_vs_n_s_low_SNR}). 
        
        In Fig.~\ref{fig:Pe_vs_R_CE_high_SNR} and Fig.~\ref{fig:Pe_vs_R_CE_low_SNR}, we plot the probability of error versus $R_\textnormal{CE}$, with $\alpha^2=2$ and $R_\textnormal{CA}=1$ fixed. For the homodyne detector and the Helstrom bound, we plot the theoretical probability of error, without considering dark counts, as a horizontal line as they do not depend on $R_\textnormal{CE}$. We choose the SNR  $R_\textnormal{SN}=10^6$ in Fig.~\ref{fig:Pe_vs_R_CE_high_SNR} and $R_\textnormal{SN}=10^2$ in Fig.~\ref{fig:Pe_vs_R_CE_low_SNR}. Similarly, we see that when the SNR is ``high" (Fig.~\ref{fig:Pe_vs_n_s_high_SNR}), our method outperforms the homodyne detector when $R_\textnormal{CE}\gtrsim 0.67$, while when the SNR is ``low" (Fig.~\ref{fig:Pe_vs_n_s_low_SNR}), our method outperforms the homodyne detector when $R_\textnormal{CE}\gtrsim 0.92$.

        Fig.~\ref{fig:Pe_vs_R_CE_high_SNR} and Fig.~\ref{fig:Pe_vs_R_CE_low_SNR} also agree with Proposition~\ref{prop:policy}. In the former, when SNR is ``high", the straight line indicates exponent-optimal time-sharing policies between displacement $0$ and $\alpha$; in the latter, when SNR is ``low", the slight curve around $R_\textnormal{CE}\in[0.9,1]$ indicates that time-sharing policies between displacement $0$ and $\alpha$ are no longer exponent-optimal.
        
        \begin{figure}
            \centering
            \includegraphics[width=0.65\linewidth]{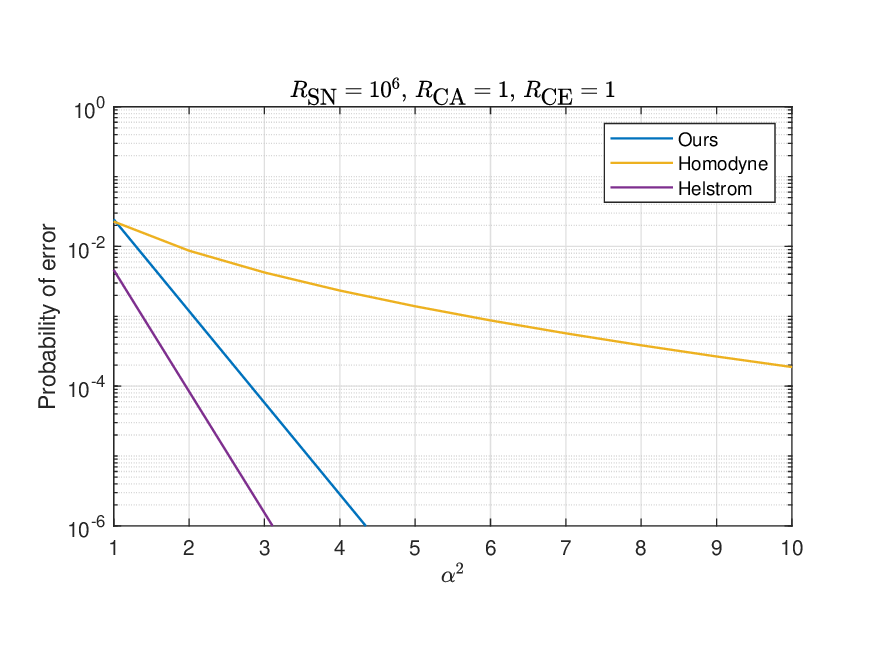}
                \caption{Probability of error versus $\alpha^2$. Parameter values are $R_\textnormal{SN}=10^{6}$ ($r_\textnormal{SN}=10^{-6}$), $R_\textnormal{CA}=1$ and $R_\textnormal{CE}=1$.
                }
                \label{fig:Pe_vs_n_s_high_SNR}
        \end{figure}

        \begin{figure}
            \centering
            \includegraphics[width=0.65\linewidth]{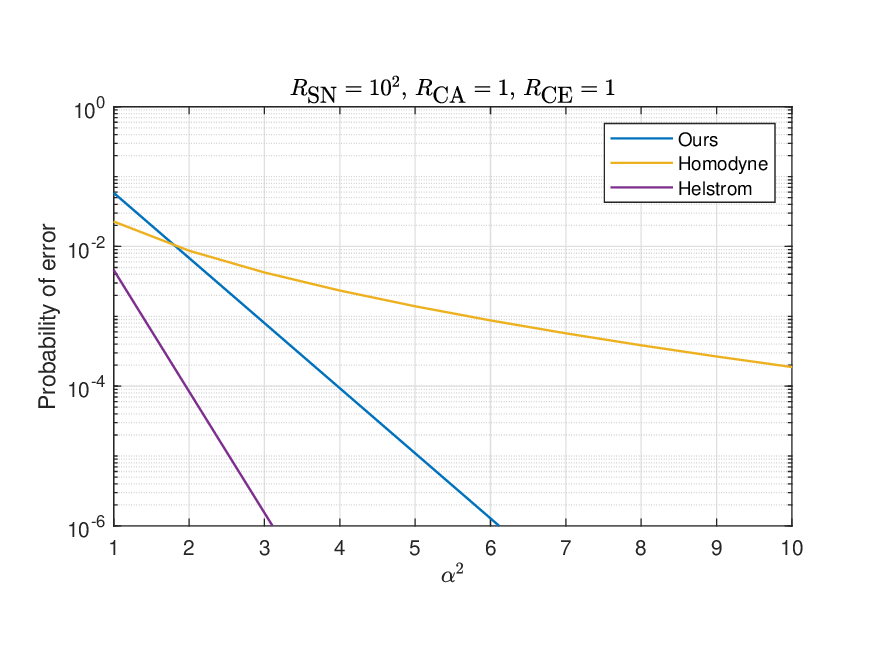}
                \caption{Probability of error versus $\alpha^2$. Parameter values are $R_\textnormal{SN}=10^{2}$ ($r_\textnormal{SN}=10^{-2}$), $R_\textnormal{CA}=1$ and $R_\textnormal{CE}=1$.
                }
                \label{fig:Pe_vs_n_s_low_SNR}
        \end{figure}
    
        \begin{figure}
            \centering
            \includegraphics[width=0.65\linewidth]{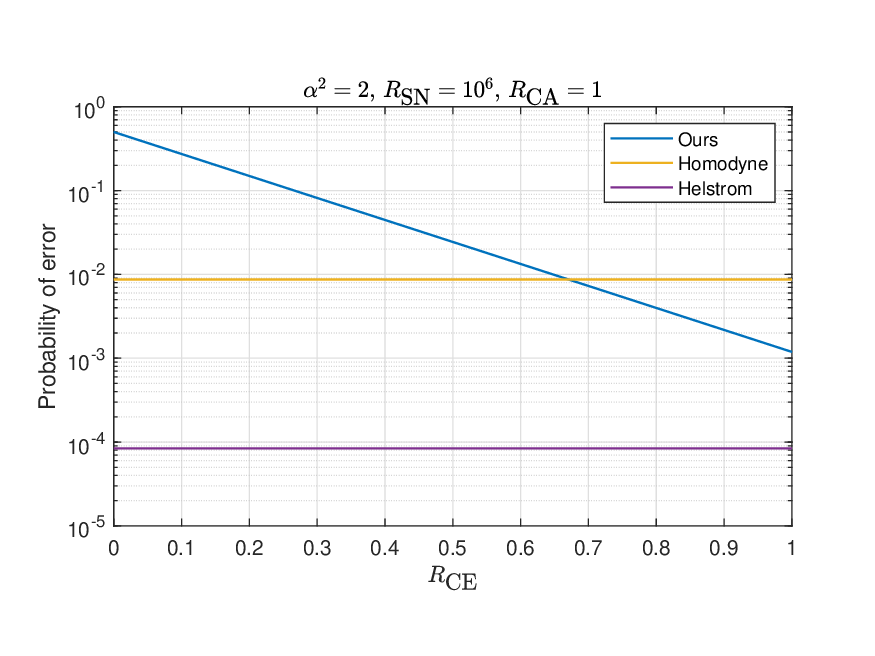}
                \caption{Probability of error versus $R_\textnormal{CE}$. Parameter values are $\alpha^2=2~(\text{photon}/\text{s})$, $R_\textnormal{SN}=10^{6}$ ($r_\textnormal{SN}=10^{-6}$), and $R_\textnormal{CA}=1$.
                }
                \label{fig:Pe_vs_R_CE_high_SNR}
        \end{figure}

        \begin{figure}
                \centering
                \includegraphics[width=0.65\linewidth]{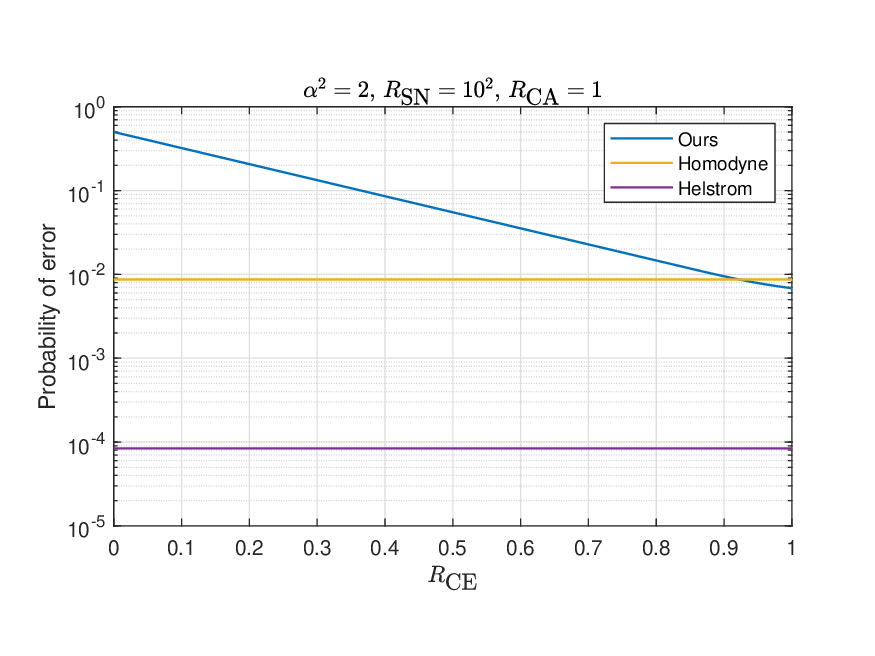}
                \caption{Probability of error versus $R_\textnormal{CE}$. Parameter values are $\alpha^2=2~(\text{photon}/\text{s})$, $R_\textnormal{SN}=10^{2}$ ($r_\textnormal{SN}=10^{-2}$), and $R_\textnormal{CA}=1$.
                }
                \label{fig:Pe_vs_R_CE_low_SNR} 
            \end{figure}

\appendices        
\section{Proof of Theorem \ref{thm:char}}
\label{appen-proof-thm}
    
    First, we give an ``achievability" proof for the exponent and then derive a closely related upper bound on the error probability. We follow the proof strategy of \cite[Theorem 1]{nitinawarat_controlled_2013} closely but also make necessary adjustments to account for the fact that in our case, the observation kernels change with the sample size $N$. Also, to obtain a valid upper bound on error probability, we take extra care in following the multiplicative constants throughout the proof.

    Fix $\alpha>0$, the operating parameters $R_\textnormal{SN},R_\textnormal{CA},R_\textnormal{CE}>0$ and $K,N\in\bbN$. Recall the definitions of $\calU_{\alpha,K}$ and $\calE_\alpha$, and we use the shorthand $p_m^u = p_m^u \left[ \alpha , \frac{\alpha^2}{R_\textnormal{SN}} , \frac{1}{N} \right].$
    
    Choose an arbitrary $ \tilde{u}^N = (\tilde{u}_1,\cdots,\tilde{u}_N) \in \calU_{\alpha,K}^N $ such that $\frac{1}{N} \sum_{n\in[N]} |\tilde{u}_n|^2 \leq \calE_\alpha$. Under the pure open-loop control policy $q_{\tilde{u}^N}$ defined by $q_{\tilde{u}^N} (u^N) \triangleq \mathbb{I}[u^N=\tilde{u}^N]$, it is well-known from statistical decision theory that the optimal decision rule $\delta:\calY^N\times\calU_{\alpha,K}^N\to\calM$ is the maximum likelihood test $\delta_\textnormal{ML}^{\tilde{u}^N}:\calY^N\to\calM$ whose form depends on $\tilde{u}^N$. Following \cite[eq. (34), (35)]{nitinawarat_controlled_2013}, we obtain that, similar to \cite[eq. (36)]{nitinawarat_controlled_2013}, for all $(\ell,m)\in\calM^2$,
    \begin{align}
        &
        \bbP_\ell \left\{ \delta_\textnormal{ML}^{\tilde{u}^N} \left(Y^N\right) = m \right\}
        +
        \bbP_m \left\{ \delta_\textnormal{ML}^{\tilde{u}^N} \left(Y^N\right) = \ell \right\}
        \nonumber
        \\
        \leq~&
        2
        \exp
        \left(
        -
        N
        \max_{s\in[0,1]}
        \sum_{u\in\calU_{\alpha,K}}
        \bar{q}_{\tilde{u}^N}(u)
        \mathbb{C}_s( p_\ell^{u} \| p_m^{u} ; \mu_{\bbN_0} )
        \right)
        ,
         \label{eq:sym_err_upp_bdd}
    \end{align}
        where $\bar{q}_{\tilde{u}^N}:\calU_{\alpha,K}\to [0,1]$ is the type (i.e., empirical distribution) of $\tilde{u}^N$ defined as 
        \begin{align*}
            \bar{q}_{\tilde{u}^N}(u) 
            \eqdef
            \frac{1}{N} \sum_{n\in[N]} \mathbb{I}[u=\tilde{u}_n]
            .
        \end{align*}
        We include an extra factor of $2$ in \eqref{eq:sym_err_upp_bdd} compared to \cite[eq. (36)]{nitinawarat_controlled_2013} because the observation kernels $p^u_m$ are discrete, hence $\bbP_\ell\{ p_\ell(Y^N)=p_m(Y^N)\}>0$ and $\bbP_m\{ p_\ell(Y^N)=p_m(Y^N)\}>0$, while the exact upper bound given in \cite[eq. (36)]{nitinawarat_controlled_2013} is valid for continuous observation kernels. Then we obtain
        \begin{align*}
            &
            \inf_{q\in\mathcal{Q}_\textnormal{OL}(\calU_{\alpha,K}^N,E,\calE_\alpha)}
            \inf_{\delta:\calY^N\times\calU_{\alpha,K}^N\to\calM}
            P_\textnormal{e}
            \left(
            \left\lbrace 
            p_m^u
            \right\rbrace^{u\in\calU_{\alpha,K}}_{m\in\calM}
            , ( q , \delta )
            \right)
            \\
            &~\leq
            \inf_{\delta:\calY^N\times\calU_{\alpha,K}^N\to\calM}
            P_\textnormal{e}
            \left( \{p_m^u\}_{m\in\calM}^{u\in\calU_{\alpha,K}}
            , 
            \left( q_{\tilde{u}^N} , \delta \right) \right)
            \\
            &~=
            P_\textnormal{e}
            \left( \{p_m^u\}_{m\in\calM}^{u\in\calU_{\alpha,K}}
            , \left( q_{\tilde{u}^N} 
            , \delta^{\tilde{u}^N}_\textnormal{ML} \right)
            \right)
            \\
            &~\eqdef
            \frac{1}{|\calM|}\sum_{m\in\calM}
            \bbP_m\left\{
            \delta^{\tilde{u}^N}_\textnormal{ML}
            \left( Y^N \right)
            \neq m
            \right\}
            \\
            &~=
            \frac{1}{|\calM|}
            \sum_{\substack{(\ell,m)\in\calM^2\\\ell<m}}
            \left(
            \bbP_\ell \left\{
            \delta^{\tilde{u}^N}_\textnormal{ML}
            \left( Y^N \right)
            = m
            \right\}
            +
            \bbP_m \left\{
            \delta^{\tilde{u}^N}_\textnormal{ML}
            \left( Y^N \right)
            = \ell
            \right\}
            \right)
            \\
            &\leq
            \frac{2}{|\calM|}
            \sum_{\substack{(\ell,m)\in\calM^2\\\ell<m}}
            \exp
            \left(
            -
            N
            \max_{s\in[0,1]}
            \sum_{u\in\calU_{\alpha,K}}
            \bar{q}_{\tilde{u}^N}(u)
            \mathbb{C}_s( p_k^{u} \| p_m^{u} ; \mu_{\bbN_0} )
            \right)
        \end{align*}
        where the penultimate equality is obtained through the symmetrization step in \cite[eq. (33)]{nitinawarat_controlled_2013}.
        
        Next, we apply properties of $p_m^u = p_m^u \left[ \alpha , \frac{\alpha^2}{R_\textnormal{SN}} , \frac{1}{N} \right] = \textnormal{Poi}( \lambda^u_m )$, where 
        \begin{align}
            \label{eq:lambda-Lambda-conver}
            \lambda^u_m
            =
            \frac{1}{N}
            \left(
            \left\vert
            u + \alpha \textnormal{e}^{i\phi_m}
            \right\vert^2
            +
            \frac{\alpha^2}{R_\textnormal{SN}}
            \right)
            =
            \frac{\alpha^2}{N}
            \Lambda_m^v
            ,
        \end{align}
        where $v=u/\alpha\in\calV_K$, to rewrite this upper bound in terms of the mean photon number $n_\textnormal{s}=\alpha^2$. In particular, since a straightforward calculation of the Chernoff $s$-divergence between Poissons via the definition yields that for any $\lambda_0,\lambda_1>0$,
        \begin{align*}
            \mathbb{C}_s
            (
            \text{Poi}(\lambda_0) \| \text{Poi}(\lambda_1) ; \mu_{\mathbb{N}_0}
            )
            &= s \lambda_0 
            + (1-s) \lambda_1
            - \lambda_0^s \lambda_1^{1-s}
            \\
            &\leq
            s \lambda_0 
            + (1-s) \lambda_1
            \leq
            \max \left\{ \lambda_0 , \lambda_1 \right\}
            ,
        \end{align*}
        we see that
        \begin{align}
            \label{eq:p-to-P-conversion}
            \mathbb{C}_s( p_\ell^{u} \| p_m^{u} ; \mu_{\bbN_0} )
            =
            \frac{\alpha^2}{N}
            \mathbb{C}_s( P_\ell^{v} \| P_m^{v} ; \mu_{\bbN_0} )
            ,
        \end{align}
        where for any $(\ell,m)\in\calM^2$ and $v\in\calV_K\subset\calD(R_\textnormal{CA})$ we have
        \begin{align}
            \begin{split} \label{eq:Cher_div_unif_bdd}
            0
            \leq
            \mathbb{C}_s( P_\ell^{v} \| P_m^{v} ; \mu_{\bbN_0} )
            &\leq
            \max \left\{
            \Lambda_\ell^v, \Lambda_m^v
            \right\}
            =
            \max \left\{
            | v + \textnormal{e}^{i\phi_\ell} |^2
            ,
            | v + \textnormal{e}^{i\phi_m} |^2
            \right\}
            + r_\textnormal{SN}
            \\
            &\leq
            (|v|+1)^2 + r_\textnormal{SN}
            \leq
            (R_\textnormal{CA}+1)^2 + r_\textnormal{SN}
            <
            \infty
            ,
            \end{split}
        \end{align}
        so that $\mathbb{C}_s( P_\ell^{v} \| P_m^{v} ; \mu_{\bbN_0} )$ is uniformly bounded for any $(\ell,m)\in\calM^2$. Applying the $\alpha^2/N$ scaling in the previously obtained upper bound, we can re-write
        \begin{align} 
            &
            \inf_{q\in\mathcal{Q}_\textnormal{OL}(\calU_{\alpha,K}^N,E,\calE_\alpha)}
            \inf_{\delta:\calY^N\times\calU_{\alpha,K}^N\to\calM}
            P_\textnormal{e}
            \left(
            \left\lbrace 
            p_m^u
            \right\rbrace^{u\in\calU_{\alpha,K}}_{m\in\calM}
            , ( q , \delta )
            \right)
            \nonumber \\
            &~\leq
            \frac{2}{|\calM|}
            \sum_{\substack{(\ell,m)\in\calM^2\\\ell<m}}
            \exp
            \left(
            -
            \alpha^2
            \max_{s\in[0,1]}
            \sum_{v\in\calV_K}
            \bar{q}_{\tilde{v}^N}(v)
            \mathbb{C}_s( P_\ell^{v} \| P_m^{v} ; \mu_{\bbN_0} )
            \right)
            \label{eq:err_upp_bdd_tight}
        \end{align}
        where $\tilde{v}^N\eqdef(\tilde{v}_1,\ldots,\tilde{v}_N)\in\calV_K^N$ with $\tilde{v}_n\eqdef\tilde{u}_n/\alpha$, and $\bar{q}_{\tilde{v}^N}(v)$ is the type of $\tilde{v}^N$ on $\calV_K$. Notice that $\bar{q}_{\tilde{v}^N}\in\calQ(\calV_K,N,R_\textnormal{CE})\eqdef\{\bar{q}:\calV_K\to\{0,1/N,\ldots,1\}\vert \sum_{v\in\calV_K} \bar{q}(v)=1 , \sum_{v\in\calV_K} |v|^2 \bar{q}(v)\leq R_\textnormal{CE}\}$. To get a lower bound on the exponent, we can follow \cite{nitinawarat_controlled_2013} to bound each exponential decay term in \eqref{eq:err_upp_bdd_tight} by the largest one, resulting in the bound
        \begin{align}
             &
             \inf_{q\in\mathcal{Q}_\textnormal{OL}(\calU_{\alpha,K}^N,E,\calE_\alpha)}
            \inf_{\delta:\calY^N\times\calU_{\alpha,K}^N\to\calM}
            P_\textnormal{e}
            \left(
            \left\lbrace 
            p_m^u
            \right\rbrace^{u\in\calU_{\alpha,K}}_{m\in\calM}
            , ( q , \delta )
            \right)
            \nonumber \\
            &~\leq
            (|\calM|-1)
            \exp
            \left(
            -
            \alpha^2
            \min_{\substack{(\ell,m)\in\calM^2\\\ell<m}}
            \max_{s\in[0,1]}
            \bbE_{V\sim \bar{q}_{\tilde{v}^N}}
            \left[
            \mathbb{C}_s( P_\ell^{V} \| P_m^{V} ; \mu_{\bbN_0} )
            \right]
            \right)
            .
            \label{eq:err_upp_bdd_expon}
        \end{align}
        This probability of error upper bound further results in
        \begin{align*}
            &
            - \frac{1}{\alpha^2}
             \inf_{q\in\mathcal{Q}_\textnormal{OL}(\calU_{\alpha,K}^N,E,\calE_\alpha)}
            \inf_{\delta:\calY^N\times\calU_{\alpha,K}^N\to\calM}
            \log
            P_\textnormal{e}
            \left(
            \left\lbrace 
            p_m^u
            \right\rbrace^{u\in\calU_{\alpha,K}}_{m\in\calM}
            , ( q , \delta )
            \right)
            \\
            &~\geq
            -\frac{\log(|\calM|-1)}{\alpha^2}
            +
            \min_{\substack{(\ell,m)\in\calM^2\\\ell<m}}
            \max_{s\in[0,1]}
            \bbE_{V\sim \bar{q}_{\tilde{v}^N}}
            \left[
            \mathbb{C}_s( P_\ell^{V} \| P_m^{V} ; \mu_{\bbN_0} )
            \right]
            ,
        \end{align*}
        where the second term on the right-hand side does not depend on $\alpha$. Hence, we get
        \begin{align}
            \begin{split} \label{eq:achiev_last_step}
            &- 
            \liminf_{\alpha\to\infty}
            \inf_{q\in\mathcal{Q}_\textnormal{OL}(\calU_{\alpha,K}^N,E,\calE_\alpha)}
            \inf_{\delta:\calY^N\times\calU_{\alpha,K}^N\to\calM}
            \frac{1}{\alpha^2}
            \log
            P_\textnormal{e}
            \left(
            \left\lbrace 
            p_m^u
            \right\rbrace^{u\in\calU_{\alpha,K}}_{m\in\calM}
            , ( q , \delta )
            \right)
            \\
            \geq~&
            \min_{\substack{(\ell,m)\in\calM^2\\\ell<m}}
            \max_{s\in[0,1]}
            \bbE_{V\sim \bar{q}_{\tilde{v}^N}}
            \left[
            \mathbb{C}_s( P_\ell^{V} \| P_m^{V} ; \mu_{\bbN_0} )
            \right]
            .
            \end{split}
        \end{align}

        For any $(\ell,m)\in\calM^2$, from its explicit expression, $\mathbb{C}_s( P_\ell^{v} \| P_m^{v} ; \mu_{\bbN_0} )$ can be seen to be continuous on $v\in\calD(R_\textnormal{CA})$, and hence it is \emph{uniformly} continuous on the compact set $\calD(R_\textnormal{CA})$. Also, by construction, $\calV_K$ is an $\Theta(1/K)$-net of $\calD(R_\textnormal{CA})$. Recall that $\tilde{u}^N\in\calU_{\alpha,K}^N$, and hence $\tilde{v}^N\in\calV_K^N$, is chosen arbitrarily at the start. Hence, taking limit superior on both sides of \eqref{eq:achiev_last_step}, we obtain
        \begin{align*}
            \beta_\textnormal{OL}
            \geq
            \min_{\substack{(\ell,m)\in\calM^2\\\ell<m}}
            \max_{s\in[0,1]}
            \bbE_{V\sim Q}
            \left[
            \mathbb{C}_s( P_\ell^{V} \| P_m^{V} ; \mu_{\bbN_0} )
            \right]
            ,
        \end{align*}
        for \emph{any} $Q\in\calQ(R_\textnormal{CA},R_\textnormal{CE})$. Taking supremum over $Q\in\calQ(R_\textnormal{CA},R_\textnormal{CE})$ further results in the desired lower bound:        
        \begin{align*}
            \beta_\textnormal{OL}
            \geq
            \sup_{Q\in\calQ(R_\textnormal{CA},R_\textnormal{CE})}
            \min_{\substack{(\ell,m)\in\calM^2\\\ell<m}}
            \max_{s\in[0,1]}
            \bbE_{V\sim Q}
            \left[
            \mathbb{C}_s( P_\ell^{V} \| P_m^{V} ; \mu_{\bbN_0} )
            \right]
            .
        \end{align*}
        
        Similarly, taking limit inferior as $(K,N)\to(\infty,\infty)$ on both sides of \eqref{eq:err_upp_bdd_expon}, we obtain 
        \begin{align*}
            &
            \liminf_{K,N\to\infty}
             \inf_{q\in\mathcal{Q}_\textnormal{OL}(\calU_{\alpha,K}^N,E,\calE_\alpha)}
            \inf_{\delta:\calY^N\times\calU_{\alpha,K}^N\to\calM}
            P_\textnormal{e}
            \left(
            \left\lbrace 
            p_m^u
            \right\rbrace^{u\in\calU_{\alpha,K}}_{m\in\calM}
            , ( q , \delta )
            \right)
            \\
            &~\leq
            (|\calM|-1)
            \exp
            \left(
            -
            \alpha^2
            \sup_{Q\in\calQ(R_\textnormal{CA},R_\textnormal{CE})}
            \min_{\substack{(\ell,m)\in\calM^2\\\ell<m}}
            \max_{s\in[0,1]}
            \bbE_{V\sim Q}
            \left[
            \mathbb{C}_s( P_k^{V} \| P_m^{V} ; \mu_{\bbN_0} )
            \right]
            \right)
            ,
        \end{align*}
        which, after we pin down the exponent by the following ``converse" proof, is the desired upper bound on the probability of error.
        
        Next, we give the converse proof that provides a matching upper bound for the exponent. We adapt the proof strategy of converse in \cite[Theorem 1]{nitinawarat_controlled_2013} to our case. In particular, since the last parameter of our observation kernels $\{p_m^u\left[\alpha,\alpha^2/R_\textnormal{SN},1/N\right]\}$ shrinks with the sample size $N$, it is \emph{a priori} unclear whether the argument in \cite[Theorem 1]{nitinawarat_controlled_2013} still results in the exponent we define.

        Fix $\alpha>0$, the operating parameters $R_\textnormal{SN},R_\textnormal{CA},R_\textnormal{CE}>0$ and $K,N\in\bbN$. Recall the definitions of $\calU_{\alpha,K}$ and $\calE_\alpha$, and we use the shorthand $p_m^u = p_m^u \left[ \alpha , \frac{\alpha^2}{R_\textnormal{SN}} , \frac{1}{N} \right].$

        Recall that any open-loop control policy $q = \{ q_n( u_n | u^{n-1}) \}_{n\in[N]} $ can be equivalently characterized by the joint PMF as $q=q(u^N)=\prod_{n\in [N]} q_n( u_n | u^{n-1})$. By definition of infimum, there is a sequence of tests $(q^{(j)},\delta^{(j)})_{j\in\bbN}$ such that
        \begin{align*}
            &
            \lim_{j\to\infty}
            P_\textnormal{e}
            \left(
            \left\lbrace p_m^u \left[ \alpha , \frac{\alpha^2}{R_\textnormal{SN}} , \frac{1}{N} \right]
            \right\rbrace^{u\in\calU_{\alpha,K}}_{m\in\mathcal{M}}
            , ( q^{(j)} , \delta^{(j)} )
            \right)
            \\=~&
            \inf_{q\in\mathcal{Q}_\textnormal{OL}(\calU_{\alpha,K}^N,E,\calE_\alpha)}
            \inf_{\delta:\calY^N\times\calU_{\alpha,K}^N\to\calM}
            P_\textnormal{e}
            \left(
            \left\lbrace 
            p_m^u
            \right\rbrace^{u\in\calU_{\alpha,K}}_{m\in\calM}
            , ( q , \delta )
            \right)
            .
        \end{align*}
        
        Fix any $j\in\bbN$. Sample a control sequence $U^N\sim q^{(j)}$, and denote the realization $u^N\in\calU_{\alpha,K}^N$. Define $v^N=u^N/\alpha,$ so $v^N\in\calU_{K}^N$.
        
        Fix any pair of hypotheses $(\ell,m)\in\calM^2$ with $\ell\neq m$. Then define
        \begin{align*}
            s^\star
            \eqdef
            \begin{cases}
                \argmax_{s\in[0,1]} \sum_{n\in [N]} 
                \bbC_s \left( p_{\ell}^{u_n} \| p_m^{u_n} ; \mu_{\bbN_0} \right) , & \mbox{if singleton,}
                \\
                \frac{1}{2} , & \mbox{otherwise.}
            \end{cases}
        \end{align*}
        From the strict convexity of Chernoff $s$-divergences w.r.t. $s$ between any distinct distributions, the argmax set is a singleton as soon as there is any $u_n$         such that $p_{\ell}^{u_n} \neq p_m^{u_n}$; on the other hand, if $p_{\ell}^{u_n} = p_m^{u_n}$ for all $n\in[N]$, we must have $\sum_{n\in [N]} \bbC_s \left( p_{\ell}^{u_n} \| p_m^{u_n} ; \mu_{u_n} \right)=0.$ Moreover, applying \eqref{eq:p-to-P-conversion}, we can rewrite
        \begin{align*}
            s^\star
            \eqdef
            \begin{cases}
                \argmax_{s\in[0,1]} \bbE_{V\sim\bar{q}_{v^N}}
                \left[ \bbC_s \left( P_{\ell}^{V} \| P_m^{V} ; \mu_{\bbN_0} \right) \right] , & \mbox{if singleton,}
                \\
                \frac{1}{2} , & \mbox{otherwise,}
            \end{cases}
        \end{align*}
        where $\bar{q}_{v^N}$ is the type of $v^N$. Hence we see that $s^\star$ only depends on $(\ell,m)$ and the type $\bar{q}_{v^N}$.
        
        For each $n\in[N]$, we further define the $s^\star$-tilted distribution between $p_\ell^{u_n}$ and $p_m^{u_n}$ via its density
        \begin{align}
            b_{\ell,m}^{u_n}(y)
            &\triangleq
            \frac{ [ p_\ell^{u_n} (y) ]^{s^\star} [ p_m^{u_n}(y) ]^{1-s^\star} }{ \int_{\calY} [ p_\ell^{u_n} (y) ]^{s^\star} [ p_m^{u_n}(y) ]^{1-s^\star} \mathrm{d}\mu_{u_n}(y) }
            \nonumber\\
            &=
            \exp\left( \bbC_{s^\star}\left( p_\ell^{u_n} \| p_m^{u_n} ; \mu_{u_n} \right)\right)
            [ p_\ell^{u_n} (y) ]^{s^\star} [ p_m^{u_n}(y) ]^{1-s^\star}
            \label{eq:univar-tilted-distri}
        \end{align}
        w.r.t. $\mu_{u_n}$. We also define an overall tilted distribution on $\calY^N$ via its density
        \begin{align*}
            b_{\ell,m}\left(y^N\right)
            \triangleq
            \prod_{n\in[N]} b_{\ell,m}^{u_n}(y_n)
        \end{align*} 
        w.r.t. $\prod_{n\in[N]}\mu_{u_n}$. Equivalently, this overall tilted distribution is a probability measure $\tilde{\bbP}_{\ell,m}$ where $$\tilde{\bbP}_{\ell,m}\left\{Y^N\in\calA\right\}\eqdef\int_\calA b_{\ell,m}\left(y^N\right)\prod_{n\in[N]}\mu_{u_n}(y_n)$$ for $\calA\in\mathscr{F}_{\calY^N}$, and we denote the associated expectation by $\tilde{\bbE}_{\ell,m}$.
        
        Since the control policy is open-loop, and now $u^N$ is fixed, we have that $(Y_n)_{n\in[N]}$ has \emph{independent} components conditioned on any hypotheses $m$, i.e., $Y^N|(U^N=u^N) \sim \prod_{n\in[N]} p_m^{u_n}(y_n)$ if $M=m$. Then we see that $b_{\ell,m}$ is the $s^\star$-tilted distribution between $\prod_{n\in[N]} p_\ell^{u_n}$ and $\prod_{n\in[N]} p_m^{u_n}$, and hence 
        \begin{align*}
            \bbD \left( b_{\ell,m} \Bigg\| \prod_{n\in[N]} p_\ell^{u_n} ; \prod_{n\in[N]}\mu_{u_n} \right)
            &=
            \bbD \left( b_{\ell,m} \Bigg\| \prod_{n\in[N]} p_m^{u_n} ; \prod_{n\in[N]}\mu_{u_n} \right)
            \\
            &=
            \bbC_{s^\star} \left( \prod_{n\in[N]} p_\ell^{u_n} \Bigg\| \prod_{n\in[N]} p_m^{u_n} ; \prod_{n\in[N]}\mu_{u_n} \right)
            ,
        \end{align*}
        or equivalently,
        \begin{align}
            \label{eq:opt-Chernoff-exponent}
            \sum_{n\in[N]} 
            \bbD \left( b_{\ell,m}^{u_n} \| p_{\ell}^{u_n} ; \mu_{u_n}
            \right)
            =
            \sum_{n\in[N]} 
            \bbD \left( b_{\ell,m}^{u_n} \| p_m^{u_n} ; \mu_{u_n}
            \right)
            =
            \sum_{n\in[N]} 
            \bbC_{s^\star} \left( p_{\ell}^{u_n} \| p_m^{u_n} ; \mu_{u_n}
            \right)
            ,
        \end{align}
        which is \cite[eq. (39)]{nitinawarat_controlled_2013}.
        
        Now we focus on the following r.v.
        \begin{align}
            S_N 
            &\eqdef
            \log 
            \left(
            \frac{\prod_{n\in[N]} b_{\ell,m}^{u_n} (Y^N)}
            {\prod_{n\in[N]} p_{m}^{u_n} (Y^N)}
            \right)
            -
            \tilde{\bbE}_{\ell,m}
            \left[
            \log 
            \left(
            \frac{\prod_{n\in[N]} b_{\ell,m}^{u_n} (Y^N)}
            {\prod_{n\in[N]} p_{m}^{u_n} (Y^N)}
            \right)
            \right]
            \nonumber\\&=
            \sum_{n\in[N]}
            \left(
            \log \left( \frac{ b_{\ell,m}^{u_n}(Y_n) }{ p_{m}^{u_n}(Y_n) } \right)
            -
            \tilde{\bbE}_{\ell,m}
            \left[
            \log \left( \frac{ b_{\ell,m}^{u_n}(Y_n) }{ p_{m}^{u_n}(Y_n) } \right)
            \right]
            \right)
            \nonumber\\&=
            \log 
            \left(
            \frac{\prod_{n\in[N]} b_{\ell,m}^{u_n} (Y^N)}
            {\prod_{n\in[N]} p_{m}^{u_n} (Y^N)}
            \right)
            -
            \sum_{n\in[N]}
            \bbD \left( b_{\ell,m}^{u_n} \| p_m^{u_n} ; \mu_{u_n} \right)
            \label{eq:sum-center-log-like}
            ,
        \end{align}
        which by definition has mean zero under $\tilde{\bbP}_{\ell,m}$. To obtain a concentration result, we further consider its variance $\widetilde{\textnormal{Var}}_{\ell,m}(S_N)$, where for any r.v. $Z$,
        \begin{align*}
            \widetilde{\textnormal{Var}}_{\ell,m}(Z)
            \eqdef \tilde{\bbE}_{\ell,m}
            \left[ \left( Z - \tilde{\bbE}_{\ell,m}\left[Z\right] \right) ^2 \right]
            .
        \end{align*}
        First, we note that by \eqref{eq:univar-tilted-distri},
        \begin{align*}
            \log \left( \frac{ b_{\ell,m}^{u_n}(y) }{ p_{m}^{u_n}(y) } \right)
            = 
            s^\star \log \left( \frac{ p_{\ell}^{u_n}(y) }{ p_{m}^{u_n}(y) } \right)
            +
            \bbC_{s^\star} \left( p_{\ell}^{u_n} \| p_m^{u_n} ; \mu_{u_n}
            \right)
            ,
        \end{align*}
        so
        \begin{align*}
            S_N &=
            s^\star
            \sum_{n\in[N]}
            \left(
            \log \left( \frac{ p_{\ell}^{u_n}(Y_n) }{ p_{m}^{u_n}(Y_n) } \right)
            -
            \tilde{\bbE}_{\ell,m}
            \left[
            \log \left( \frac{ p_{\ell}^{u_n}(Y_n) }{ p_{m}^{u_n}(Y_n) } \right)
            \right]
            \right)
            ,
        \end{align*}
        Since $(Y_n)_{n\in[N]}$ independent under $\tilde{\bbP}_{\ell,m}$, we have
        \begin{align*}
            \widetilde{\textnormal{Var}}_{\ell,m}(S_N)
            = (s^\star)^2 \sum_{n\in[N]}
            \widetilde{\textnormal{Var}}_{\ell,m}
            \left( \log \left( \frac{ p_{\ell}^{u_n}(Y_n) }{ p_{m}^{u_n}(Y_n) } \right) \right)
            .
        \end{align*}
        Explicitly calculating the log-likelihood ratios for the Poisson observation kernels gives
        \begin{align*}
            \log \left( \frac{ p_{\ell}^{u_n}(y) }{ p_{m}^{u_n}(y)} \right)
            &= y \log \left( \frac{ \lambda_{\ell}^{u_n} }{ \lambda_{m}^{u_n}} \right)
            - \left( \lambda_{\ell}^{u_n} - \lambda_{m}^{u_n} \right)
            .
        \end{align*}
        Moreover, under $\tilde{\bbP}_{\ell,m}$ we have $Y_n\sim b_{\ell,m}^{u_n}$, while a direct calculation via \eqref{eq:univar-tilted-distri} yields that $b_{\ell,m}^{u_n}$ is the density of $\textnormal{Poi}\left( (\lambda_\ell^{u_n})^{s^\star} (\lambda_m^{u_n})^{1-s^\star} \right)$ w.r.t. $\mu_{\bbN_0}$. Hence, 
        \begin{align}
            \widetilde{\textnormal{Var}}_{\ell,m}(S_N)
            &= (s^\star)^2 \sum_{n\in[N]}
            \left(
            \log \left( \frac{ \lambda_{\ell}^{u_n} }{ \lambda_{m}^{u_n}} \right) 
            \right)^2
            \widetilde{\textnormal{Var}}_{\ell,m}
            \left( Y_n \right)
            \nonumber\\
            &= (s^\star)^2 \sum_{n\in[N]}
            \left(
            \log \left( \frac{ \lambda_{\ell}^{u_n} }{ \lambda_{m}^{u_n}} \right) 
            \right)^2
            (\lambda_\ell^{u_n})^{s^\star} (\lambda_m^{u_n})^{1-s^\star}
            \nonumber\\
            &\overset{\textnormal{(a)}}{=} (s^\star)^2
            \frac{\alpha^2}{N}
            \sum_{n\in[N]}
            \left( \log \left( \frac{ \Lambda_{\ell}^{v_n} }{ \Lambda_{m}^{v_n}} \right) \right)^2
            (\Lambda_\ell^{v_n})^{s^\star} (\Lambda_m^{v_n})^{1-s^\star}
            \nonumber\\
            &= 
            \alpha^2 
            \bbE_{ V \sim \bar{q}_{v^N} }
            \left[
            (s^\star)^2
            \left( \log \left( \frac{ \Lambda_{\ell}^{V} }{ \Lambda_{m}^{V}} \right) \right)^2
            \left(\Lambda_\ell^{V}\right)^{s^\star} \left(\Lambda_m^{V}\right)^{1-s^\star}
            \right]
            \nonumber\\
            &\overset{\textnormal{(b)}}{\leq} 
            \alpha^2 
            \left( \log \left( \frac{ ( R_\textnormal{CA} + 1 )^2 + r_\textnormal{SN} }{ r_\textnormal{SN} } \right) \right)^2
            ( ( R_\textnormal{CA} + 1 )^2 + r_\textnormal{SN} )
            \label{eq:var-upp-bdd}
            ,
        \end{align}
        where in (a) we applied \eqref{eq:lambda-Lambda-conver} and in (b) we used $0\leq s^\star\leq 1$ and the bound
        \begin{align*}
            0 < r_\textnormal{SN}
            \leq \Lambda_m^{v}
            \leq ( R_\textnormal{CA} + 1 )^2 + r_\textnormal{SN}
            < \infty
        \end{align*}
        that holds for any $v\in\calD(R_\textnormal{CA})$ and any $m\in\calM$. 
        
        Fix any function $f:\bbN\to\bbR_{>0}$ satisfying $\lim_{N\to\infty} f(N)=\infty$. For any $\eta>0$, we have
        \begin{align*}
            \tilde{\bbP}_{\ell,m}
            \left\{
            \frac{|S_N|}{f(N)\alpha} \geq \eta
            \right\}
            \leq
            \frac{
            \widetilde{\textnormal{Var}}_{\ell,m}
            \left( \frac{S_N}{f(N)\alpha } \right)
            }{ \eta^2 }
            \leq
            \left(
            \frac{
            \log \left( \frac{ ( R_\textnormal{CA} + 1 )^2 + r_\textnormal{SN} }{ r_\textnormal{SN} } \right)
            \sqrt{  R_\textnormal{CA} + 1 )^2 + r_\textnormal{SN} } 
            }{ f(N) \eta }
            \right)^2
        \end{align*}
        by Chebyshev's inequality. Then, for any $\varepsilon\in(0,1/2)$, since $f(N)\to\infty$ as $N\to\infty$, there is an $N(\eta,\varepsilon,R_\textnormal{SN},R_\textnormal{CA})\in\bbN$ such that $N\geq N(\eta,\varepsilon,R_\textnormal{SN},R_\textnormal{CA})$ implies that
        \begin{align*}
            f(N) \geq
            \frac{
            \log \left( \frac{ ( R_\textnormal{CA} + 1 )^2 + r_\textnormal{SN} }{ r_\textnormal{SN} } \right)
            \sqrt{  R_\textnormal{CA} + 1 )^2 + r_\textnormal{SN} } 
            }{ \sqrt{\varepsilon} \eta }
            .
        \end{align*}
        Hence, we obtain
        \begin{align}
            \label{eq:concentration-result}
            \tilde{\bbP}_{\ell,m}
            \left\{
            \frac{|S_N|}{f(N)\alpha} \geq \eta
            \right\}
            \leq \varepsilon
            ,~\forall N \geq N(\eta,\varepsilon,R_\textnormal{SN},R_\textnormal{CA})
            .
        \end{align}
        It is important to note that the integer $N(\eta,\varepsilon,R_\textnormal{SN},R_\textnormal{CA})$ does not depend on $\alpha$.
        
        To connect this to the probability of error, consider the following two cases: In case 1, we assume
        \begin{align}
            \label{eq:case-1}
            \tilde{\bbP}_{\ell,m} 
            \left\{ 
            \delta^{(j)}( Y^N , u^N ) = \ell
            \right\}
            \geq \frac{1}{2}
            ;
        \end{align}
        in case 2, we assume
        \begin{align}
            \label{eq:case-2}
            \tilde{\bbP}_{\ell,m} 
            \left\{ 
            \delta^{(j)}( Y^N , u^N ) \neq \ell
            \right\}
            \geq \frac{1}{2}
            .
        \end{align}
        At least one of the cases must be true since $\tilde{\bbP}_{\ell,m} 
            \left\{ 
            \delta^{(j)}( Y^N , u^N ) = \ell
            \right\}+ \tilde{\bbP}_{\ell,m} 
            \left\{ 
            \delta^{(j)}( Y^N , u^N ) \neq \ell
            \right\} = 1$.
        
        Suppose we are in case 1. Then monotonicity of probability and \eqref{eq:concentration-result} gives
        \begin{align}
            \tilde{\bbP}_{\ell,m} 
            \left\{ 
            \delta^{(j)}( Y^N , u^N ) = \ell
            ,
            \frac{S_N}{f(N)\alpha} \geq \eta
            \right\}
            &\leq
            \tilde{\bbP}_{\ell,m} 
            \left\{ 
            \frac{S_N}{f(N)\alpha} \geq \eta
            \right\}
            \leq
            \tilde{\bbP}_{\ell,m} 
            \left\{ 
            \frac{|S_N|}{f(N)\alpha} \geq \eta
            \right\}
            \leq \varepsilon
            ,
            \label{eq:case-1-concen-result}
        \end{align}
        and hence through a change of measure argument we have
        \begin{align*}
            \frac{1}{2} - \varepsilon
            &\overset{\textnormal{(a)}}{\leq}
            \frac{1}{2} - \tilde{\bbP}_{\ell,m} 
            \left\{ 
            \delta^{(j)}( Y^N , u^N ) = \ell
            ,
            \frac{S_N}{f(N)\alpha} \geq \eta
            \right\}
            \\
            &\overset{\textnormal{(b)}}{\leq}
            \tilde{\bbP}_{\ell,m} 
            \left\{ 
            \delta^{(j)}( Y^N , u^N ) = \ell
            ,
            \frac{S_N}{f(N)\alpha} < \eta
            \right\}
            \\
            &\overset{\textnormal{(c)}}{=}
            \tilde{\bbP}_{\ell,m} 
            \left\{ 
            \delta^{(j)}( Y^N , u^N ) = \ell
            ,
            \log 
            \left(
            \frac{\prod_{n\in[N]} b_{\ell,m}^{u_n} (Y^N)}
            {\prod_{n\in[N]} p_{m}^{u_n} (Y^N)}
            \right)
            <
            \sum_{n\in[N]}
            \bbD \left( b_{\ell,m}^{u_n} \| p_m^{u_n} ; \mu_{u_n} \right)
            +
            f(N) \alpha \eta
            \right\}
            \\&=
            \int_{\calY^N}
            \left( \prod_{n\in[N]} b_{\ell,m}^{u_n} (y^N) \mathrm{d}\mu_{u_n}(y_n) \right)
            \bbI
            \Bigg\{ 
            \delta^{(j)}( y^N , u^N ) = \ell
            ,
            \prod_{n\in[N]} b_{\ell,m}^{u_n} (y^N)
            <
            \\
            &\qquad\qquad\qquad\qquad\qquad\qquad\quad
            \prod_{n\in[N]} p_{m}^{u_n} (y^N)
            \exp\left(
            \sum_{n\in[N]}
            \bbD \left( b_{\ell,m}^{u_n} \| p_m^{u_n} ; \mu_{u_n} \right)
            +
            f(N) \alpha \eta
            \right)
            \Bigg\}
            \\&\overset{\textnormal{(d)}}{\leq}
            \exp\left(
            \sum_{n\in[N]}
            \bbD \left( b_{\ell,m}^{u_n} \| p_m^{u_n} ; \mu_{u_n} \right)
            +
            f(N) \alpha \eta
            \right)
            \int_{\calY^N}
            \left( \prod_{n\in[N]} p_{m}^{u_n} (y^N) \mathrm{d}\mu_{u_n}(y_n) \right)
            \\
            &\qquad
            \bbI
            \Bigg\{ 
            \delta^{(j)}( y^N , u^N ) = \ell
            ,
            \prod_{n\in[N]} b_{\ell,m}^{u_n} (y^N)
            <
            \prod_{n\in[N]} p_{m}^{u_n} (y^N)
            \exp\left(
            \sum_{n\in[N]}
            \bbD \left( b_{\ell,m}^{u_n} \| p_m^{u_n} ; \mu_{u_n} \right)
            +
            f(N) \alpha \eta
            \right)
            \Bigg\}
            \\&\overset{\textnormal{(e)}}{\leq}
            \exp\left(
            \sum_{n\in[N]}
            \bbD \left( b_{\ell,m}^{u_n} \| p_m^{u_n} ; \mu_{u_n} \right)
            +
            f(N) \alpha \eta
            \right)
            \bbP_m \left\{
            \delta^{(j)}( Y^N , U^N ) = \ell
            ~\bigg\vert~ U^N = u^N
            \right\}
            \\&\overset{\textnormal{(f)}}{\leq}
            \exp\left(
            \sum_{n\in[N]}
            \bbD \left( b_{\ell,m}^{u_n} \| p_m^{u_n} ; \mu_{u_n} \right)
            +
            f(N) \alpha \eta
            \right)
            \bbP_m \left\{
            \delta^{(j)}( Y^N , U^N ) \neq m
            ~\bigg\vert~ U^N = u^N
            \right\}
            ,
        \end{align*}
        where (a) follows from \eqref{eq:case-1-concen-result}, (b) follows from \eqref{eq:case-1}, (c) follows from \eqref{eq:sum-center-log-like}, (d) is the change of measure step, and (e), (f) follows from the definition in \eqref{eq:def-joint-obser-contr-meas} and the monotonicity of probability. In summary, we obtain
        \begin{align}
            \label{eq:case-1-prob-low-bdd}
            \bbP_m \left\{
            \delta^{(j)}( Y^N , U^N ) \neq m
            ~\bigg\vert~ U^N = u^N
            \right\}
            &\geq
            \left( \frac{1}{2} - \varepsilon \right)
            \exp\left(
            -
            \sum_{n\in[N]}
            \bbD \left( b_{\ell,m}^{u_n} \| p_m^{u_n} ; \mu_{u_n} \right)
            -
            f(N) \alpha \eta
            \right)
        \end{align}
        in case 1.

        In case 2., by considering instead $\tilde{\bbP}_{\ell,m}$-centered log-likelihoods of $b_{u,m}$ versus $p_\ell$, a similar argument yields
        \begin{align}
            \label{eq:case-2-prob-low-bdd}
            \bbP_\ell \left\{
            \delta^{(j)}( Y^N , U^N ) \neq \ell
            ~\bigg\vert~ U^N = u^N
            \right\}
            &\geq
            \left( \frac{1}{2} - \varepsilon \right)
            \exp\left(
            -
            \sum_{n\in[N]}
            \bbD \left( b_{\ell,m}^{u_n} \| p_\ell^{u_n} ; \mu_{u_n} \right)
            -
            f(N) \alpha \eta
            \right)
            .
        \end{align}

        Since at least one of the cases must happen, we have, in all cases,
        \begin{align}
            &
            \max \left(
            \bbP_m \left\{
            \delta^{(j)}( Y^N , U^N ) \neq m
            ~\bigg\vert~ U^N = u^N
            \right\}
            ,
            \bbP_\ell \left\{
            \delta^{(j)}( Y^N , U^N ) \neq \ell
            ~\bigg\vert~ U^N = u^N
            \right\}
            \right)
            \nonumber \\
            \overset{\textnormal{(a)}}{\geq}~&
            \left( \frac{1}{2} - \varepsilon \right)
            \exp\left(
            -
            \sum_{n\in[N]}
            \bbC_{s^\star} \left( p_{\ell}^{u_n} \| p_m^{u_n} ; \mu_{\bbN_0}
            \right)
            -
            f(N) \alpha \eta
            \right)
            \nonumber \\
            \overset{\textnormal{(b)}}{=}~&
            \left( \frac{1}{2} - \varepsilon \right)
            \exp\left(
            -
            \max_{s\in[0,1]}
            \sum_{n\in[N]}
            \bbC_{s} \left( p_{\ell}^{u_n} \| P_m^{u_n} ; \mu_{\bbN_0}
            \right)
            -
            f(N) \alpha \eta
            \right)
            \nonumber \\
            \overset{\textnormal{(c)}}{=}~&
            \left( \frac{1}{2} - \varepsilon \right)
            \exp\left(
            -
            \frac{\alpha^2}{N}
            \max_{s\in[0,1]}
            \sum_{n\in[N]}
            \bbC_{s} \left( P_{\ell}^{v_n} \| P_m^{v_n} ; \mu_{\bbN_0}
            \right)
            -
            f(N) \alpha \eta
            \right)
            ,
            \label{eq:sym-err-low-bdd}
        \end{align}
        where (a) follows from \eqref{eq:opt-Chernoff-exponent}, \eqref{eq:case-1-prob-low-bdd} and \eqref{eq:case-2-prob-low-bdd}, (b) follows from the definition of $s^\star$, and (c) follows from \eqref{eq:p-to-P-conversion}.
        
        Recalling that at the start we sample $U^N=u^N$ from $q^{(j)}$, from \eqref{eq:sym-err-low-bdd} we get
        \begin{align}
            &
            \max \left(
            \bbP_m \left\{
            \delta^{(j)}( Y^N , U^N ) \neq m
            \right\}
            ,
            \bbP_\ell \left\{
            \delta^{(j)}( Y^N , U^N ) \neq \ell
            \right\}
            \right)
            \nonumber \\
            \geq~&
            \left( \frac{1}{2} - \varepsilon \right)
            \bbE_{ V^N \sim \tilde{q}^{(j)} }
            \left[
            \exp\left(
            -
            \frac{\alpha^2}{N}
            \max_{s\in[0,1]}
            \sum_{n\in[N]}
            \bbC_{s} \left( P_{\ell}^{V_n} \| P_m^{V_n} ; \mu_{\bbN_0}
            \right)
            -
            f(N) \alpha \eta
            \right)
            \right]
            \nonumber \\
            \overset{\textnormal{(a)}}{\geq}~&
            \left( \frac{1}{2} - \varepsilon \right)
            \exp\left(
            -
            \frac{\alpha^2}{N}
            \bbE_{ V^N \sim \tilde{q}^{(j)} }
            \left[
            \max_{s\in[0,1]}
            \sum_{n\in[N]}
            \bbC_{s} \left( P_{\ell}^{V_n} \| P_m^{V_n} ; \mu_{\bbN_0}
            \right)
            \right]
            -
            f(N) \alpha \eta
            \right)
            \label{eq:sym-err-low-bdd-uncond}
            ,
        \end{align}
        where $\tilde{q}^{(j)}(v^N)\eqdef q^{(j)}(\alpha v^N)$ and (a) follows from Jensen's inequality (i.e., the convexity of $\exp(\cdot)$).
        
        Then, the probability of error is
        \begin{align*}
            &
            P_\textnormal{e}
            \left(
            \left\lbrace p_m^u \left[ \alpha , \frac{\alpha^2}{R_\textnormal{SN}} , \frac{1}{N} \right]
            \right\rbrace^{u\in\calU_{\alpha,K}}_{m\in\mathcal{M}}
            , ( q^{(j)} , \delta^{(j)} )
            \right)
            \eqdef
            \frac{1}{|\calM|} \sum_{k\in\calM} 
            \bbP_k \left\{ 
            \delta^{(j)} \left( Y^N , U^N \right) \neq k \right\}
            \\\geq~&
            \frac{1}{|\calM|} \max_{k\in\calM} 
            \bbP_k \left\{ 
            \delta^{(j)} \left( Y^N , U^N \right) \neq k \right\}
            \\\geq~&
            \frac{1}{|\calM|} 
            \max \left(
            \bbP_m \left\{
            \delta^{(j)}( Y^N , U^N ) \neq m
            \right\}
            ,
            \bbP_\ell \left\{
            \delta^{(j)}( Y^N , U^N ) \neq \ell
            \right\}
            \right)
            \\
            \overset{\textnormal{(a)}}{\geq}~&
            \frac{ \frac{1}{2} - \varepsilon }{ |\calM| } 
            \exp\left(
            -
            \frac{\alpha^2}{N}
            \bbE_{ V^N \sim \tilde{q}^{(j)} }
            \left[
            \max_{s\in[0,1]}
            \sum_{n\in[N]}
            \bbC_{s} \left( P_{\ell}^{V_n} \| P_m^{V_n} ; \mu_{\bbN_0}
            \right)
            \right]
            -
            f(N) \alpha \eta
            \right)
        \end{align*}
        where (a) follows from \eqref{eq:sym-err-low-bdd-uncond}. Taking logarithm on both sides then dividing by $\alpha^2$, we obtain
        \begin{align}
            \begin{split}
            \label{eq:converse_last_step}
            &
            -\frac{1}{\alpha^2}
            \log
            P_\textnormal{e}
            \left(
            \left\lbrace p_m^u \left[ \alpha , \frac{\alpha^2}{R_\textnormal{SN}} , \frac{1}{N} \right]
            \right\rbrace^{u\in\calU_{\alpha,K}}_{m\in\mathcal{M}}
            , ( q^{(j)} , \delta^{(j)} )
            \right)
            \\
            \leq~&
            - \frac{ \log( \frac{1}{2} - \varepsilon ) - \log |\calM| }{ \alpha^2 } 
            +
            \bbE_{ V^N \sim \tilde{q}^{(j)} }
            \left[
            \max_{s\in[0,1]}
            \frac{1}{N}
            \sum_{n\in[N]}
            \bbC_{s} \left( P_{\ell}^{V_n} \| P_m^{V_n} ; \mu_{\bbN_0}
            \right)
            \right]
            +
            \frac{f(N)\eta}{\alpha} 
            .
            \end{split}
        \end{align}
       According to \cite{nitinawarat_controlled_2013}, we can restrict ourselves to \emph{pure} open-loop control policies. Thus, without loss of generality, we can assume $\tilde{q}^{(j)}$ is pure, namely, there is a deterministic sequence $v^N_{(j)}$ such that $\tilde{q}^{(j)}(v^N)=\bbI[v^N=v^N_{(j)}]$. Hence \eqref{eq:converse_last_step} becomes
       \begin{align*}
            &
            -\frac{1}{\alpha^2}
            \log
            P_\textnormal{e}
            \left(
            \left\lbrace p_m^u \left[ \alpha , \frac{\alpha^2}{R_\textnormal{SN}} , \frac{1}{N} \right]
            \right\rbrace^{u\in\calU_{\alpha,K}}_{m\in\mathcal{M}}
            , ( q^{(j)} , \delta^{(j)} )
            \right)
            \\
            \leq~&
            - \frac{ \log( \frac{1}{2} - \varepsilon ) - \log |\calM| }{ \alpha^2 } 
            +
            \max_{s\in[0,1]}
            \sum_{v\in\calV_K}
            \bar{q}_{ v^N_{(j)} }(v)
            \bbC_{s} \left( P_{\ell}^{v} \| P_m^{v} ; \mu_{\bbN_0}
            \right)
            +
            \frac{f(N)\eta}{\alpha} 
            ,
        \end{align*}
        where $\bar{q}_{ v^N_{(j)} }$ is the type of $v^N_{(j)} $. Since the inequality holds for any $(\ell,m)$, we have
        \begin{align*}
            &
            -\frac{1}{\alpha^2}
            \log
            P_\textnormal{e}
            \left(
            \left\lbrace p_m^u \left[ \alpha , \frac{\alpha^2}{R_\textnormal{SN}} , \frac{1}{N} \right]
            \right\rbrace^{u\in\calU_{\alpha,K}}_{m\in\mathcal{M}}
            , ( q^{(j)} , \delta^{(j)} )
            \right)
            \\
            \leq~&
            - \frac{ \log( \frac{1}{2} - \varepsilon ) - \log |\calM| }{ \alpha^2 } 
            +
            \min_{\substack{(\ell,m)\in\calM^2\\\ell<m}}
            \max_{s\in[0,1]}
            \sum_{v\in\calV_K}
            \bar{q}_{ v^N_{(j)} }(v)
            \bbC_{s} \left( P_{\ell}^{v} \| P_m^{v} ; \mu_{\bbN_0}
            \right)
            +
            \frac{f(N)\eta}{\alpha} 
            .
        \end{align*}
        Then, since $\bar{q}_{ v^N_{(j)} }\in\calQ(R_\textnormal{CA},R_\textnormal{CE})$, we further have
        \begin{align*}
            &
            -\frac{1}{\alpha^2}
            \log
            P_\textnormal{e}
            \left(
            \left\lbrace p_m^u \left[ \alpha , \frac{\alpha^2}{R_\textnormal{SN}} , \frac{1}{N} \right]
            \right\rbrace^{u\in\calU_{\alpha,K}}_{m\in\mathcal{M}}
            , ( q^{(j)} , \delta^{(j)} )
            \right)
            \\
            \leq~&
            - \frac{ \log( \frac{1}{2} - \varepsilon ) - \log |\calM| }{ \alpha^2 } 
            +
            \sup_{Q\in\calQ(R_\textnormal{CA},R_\textnormal{CE})}
            \min_{\substack{(\ell,m)\in\calM^2\\\ell<m}}
            \max_{s\in[0,1]}
            \bbE_{V\sim Q}
            \left[
            \bbC_{s} \left( P_{\ell}^{V} \| P_m^{V} ; \mu_{\bbN_0}
            \right)
            \right]
            +
            \frac{f(N)\eta}{\alpha} 
            .
        \end{align*}
        Taking $j\to\infty$, then taking $\alpha\to\infty$, and finally taking $(K,N)\to(\infty,\infty)$ yields the desired upper bound:
        \begin{align*}
            \beta_\textnormal{OL}
            \leq
            \sup_{Q\in\calQ(R_\textnormal{CA},R_\textnormal{CE})}
            \min_{\substack{(\ell,m)\in\calM^2\\\ell<m}}
            \max_{s\in[0,1]}
            \bbE_{V\sim Q}
            \left[
            \bbC_{s} \left( P_{\ell}^{V} \| P_m^{V} ; \mu_{\bbN_0}
            \right)
            \right]
            .
        \end{align*}

\section{Proof of Proposition~\ref{prop:policy}}
\label{appen-proof-prop}

    Since the hypotheses are binary, the exponent
    \begin{align}
        \begin{split} \label{eq:opt_prob_exponent}
        \beta_\textnormal{OL}
        &=
        \sup_{Q\in\calQ(R_\textnormal{CA},R_\textnormal{CE})}
        \sup_{s\in[0,1]}
        \bbE_{V\sim Q}
        \left[
        \bbC_{s} \left( P_{0}^{V} \| P_1^{V} ; \mu_{\bbN_0}
        \right)
        \right]
        \\&=
        \sup_{s\in[0,1]}
        \sup_{Q\in\calQ(R_\textnormal{CA},R_\textnormal{CE})}
        \bbE_{V\sim Q}
        \left[
        \bbC_{s} \left( P_{0}^{V} \| P_1^{V} ; \mu_{\bbN_0}
        \right)
        \right]
        ,
        \end{split}
    \end{align}
    where the exchange of supremums is justified by the uniform boundedness of $\bbC_{s} \left( P_{0}^{v} \| P_1^{v} ; \mu_{\bbN_0}
    \right)$ w.r.t. $(s,v)$ shown in \eqref{eq:Cher_div_unif_bdd}. 
    
    Since $R_\textnormal{CA}=1$ and the symmetry of the BPSK constellation $\{\ket{\alpha},\ket{-\alpha}\}$ with $\alpha>0$, without loss of generality (see, e.g., \cite{chung_capacity_2017}) we can further restrict the support of $Q$ to $[0,\alpha]\subset \bbR$.

    \textbf{Claim 1:} For $s\in(0,1/2]$, $\mathbb{C}_s\left(P_0^v\| P_1^v;\mu_{\bbN_0}\right)$ is strictly convex w.r.t. the energy function $E = E(u) = E_0 \left(u/\alpha\right)^2 = E_0 v^2 = E(v) $ (where $E_0 \triangleq \alpha^2 = n_\text{s}$ is the energy of generating Kennedy displacement) on $u\in[0, \alpha]$ (i.e., $v\in [0,1]$), when $r_\textnormal{SN} \ll 1$.

    \textbf{Proof 1:}
One can directly compute and get the Chernoff $s$-divergence between $P_0^v$ and $P_1^v$ is
$$
\mathbb{C}_s\left(P_0^v \| P_1^v;\mu_{\bbN_0}\right)=s \Lambda_0(v)+(1-s) \Lambda_1(v)-\Lambda_0(v)^s \Lambda_1(v)^{1-s}.
$$

Fix any $s\in(0,1/2]$. Observe that we can write $$\mathbb{C}_s\left(P_0^v \| P_1^v;\mu_{\bbN_0}\right)=f_s(\Lambda_0(v),\Lambda_1(v))$$ where we define $f_s:(0,\infty)\times(0,\infty)\to[0,\infty)$ by $f_s(\Lambda_0,\Lambda_1)=s\Lambda_0+(1-s)\Lambda_1-\Lambda_0^s \Lambda_1^{1-s}$.  We can justify the co-domain by the following observation: $$f_s(\Lambda_0,\Lambda_1)=s e^{\log \Lambda_0}+(1-s) e^{\log \Lambda_1}-e^{s \log \Lambda_0+(1-s) \log \Lambda_1}\geq 0$$ from convexity of $\exp(\cdot)$ and $s\in(0,1)$; or, equivalently, that $f_s(\Lambda_0,\Lambda_1)=\mathbb{C}_s(\text{Poi}(\Lambda_0)\|\text{Poi}(\Lambda_1))\geq 0$ by the non-negativity of divergences. Since $E(v)$ is strictly increasing w.r.t. $v\in[0,1]$, we can do a change of variable and re-write 
$$
\begin{aligned}
& \Lambda_0 = \Lambda_0(E)=\left(1-\sqrt{\frac{E}{E_0}}\right)^2+r_\textnormal{SN},\\
& \Lambda_1 = \Lambda_1(E)=\left(1+\sqrt{\frac{E}{E_0}}\right)^2+r_\textnormal{SN},
\end{aligned}
$$
and hence the Chernoff $s$-divergence 
$$
\mathbb{C}_s\left(P_0^v \| P_1^v;\mu_{\bbN_0}\right)
=
f_s(\Lambda_0(E),\Lambda_1(E))
$$ 
can be viewed as a function of $E\in [0,E_0]$. An equivalent condition for its strict convexity w.r.t. $E\in\left[0, E_0\right]$ is that 
$$
\frac{\partial^2}{\partial E^2} f_s(\Lambda_0(E),\Lambda_1(E)) 
> 
0
\text { for } E \in ( 0, E_0 ).
$$ 

To examine this equivalent condition, we first compute the first-order derivative:
$$
\begin{aligned}
\frac{\partial}{\partial E} f_s\left(\Lambda_0(E), \Lambda_1(E)\right)
=
\frac{\partial f_s}{\partial \Lambda_0} \cdot \frac{\mathrm{d} \Lambda_0}{\mathrm{d} E}
+
\frac{\partial f_s}{\partial \Lambda_1} \cdot \frac{\mathrm{d} \Lambda_1}{\mathrm{d} E}
,
\end{aligned} 
$$
and then the second derivative  
$$
\begin{aligned}
 \frac{\partial^2}{\partial E^2} f_s\left(\Lambda_0(E), \Lambda_1(E)\right)
 &=
 \left[
 \frac{\partial^2 f_s}{\partial \Lambda_0^2} \cdot \left(\frac{\mathrm{d} \Lambda_0}{\mathrm{d} E}\right)^2
 +
 2
 \frac{\partial^2 f_s}{\partial \Lambda_1 \partial \Lambda_0} \cdot \frac{\mathrm{d} \Lambda_1}{\mathrm{d} E} \cdot \frac{\mathrm{d} \Lambda_0}{\mathrm{d} E}
 +
 \frac{\partial^2 f_s}{\partial \Lambda_1^2} \cdot\left(\frac{\mathrm{d} \Lambda_1}{\mathrm{d} E}\right)^2
 \right]
 \\&\quad
 +
 \left[
 \frac{\partial f_s}{\partial \Lambda_0} \cdot \frac{\mathrm{d}^2 \Lambda_0}{\mathrm{d} E^2}
 +
 \frac{\partial f_s}{\partial \Lambda_1} \cdot\frac{\mathrm{d}^2 \Lambda_1}{\mathrm{d} E^2}
 \right]
 .
\end{aligned}
$$
To explicitly compute these derivatives, we further compute
$$
\begin{aligned}
& \frac{\partial f_s}{\partial \Lambda_0}
= s \left( 1 - \Lambda_0^{s-1} \Lambda_1^{1-s} \right)
,\\
&  \frac{\partial f_s}{\partial \Lambda_1}
= (1-s) \left( 1 - \Lambda_0^{s} \Lambda_1^{-s} \right)
,\end{aligned}
$$
and hence
$$
\begin{aligned}
\frac{\partial^2 f_s}{\partial \Lambda_0^2}
&=
s(1-s) \Lambda_0^{s-2} \Lambda_1^{1-s}
,\\
\frac{\partial^2 f_s}{\partial \Lambda_1^2}
&=
s(1-s) \Lambda_0^{s} \Lambda_1^{-1-s}
,\\
\frac{\partial^2 f_s}{\partial \Lambda_0 \partial \Lambda_1}
&=
-s(1-s) \Lambda_0^{s-1} \Lambda_1^{-s}
.
\end{aligned}
$$
Also, we have the derivatives of the rates $\Lambda_0(E)$ and $\Lambda_1(E)$ w.r.t. $E$ as: 
$$
\begin{aligned}
&\frac{\mathrm{d}\Lambda_0}{\mathrm{d}E}
= \frac{1}{E_0} \left( 1 - \sqrt{\frac{E_0}{E} }\right) ;
&\frac{\mathrm{d}\Lambda_1}{\mathrm{d}E}
= \frac{1}{E_0} \left( 1 + \sqrt{\frac{E_0}{E} }\right) ;
\\
&\frac{\mathrm{d}^2 \Lambda_0}{\mathrm{d} E^2 }
= \frac{\mathrm{d}^2 \Lambda_1}{\mathrm{d} E^2}
= \frac{1}{2\sqrt{E_0 E^3}}
,
\end{aligned}
$$
and hence
$$
\begin{aligned}
\left(\frac{\mathrm{d} \Lambda_0}{\mathrm{d} E}\right)^2
&= \frac{1}{E_0 E} \left( \sqrt{\frac{E}{E_0}}-1\right)^2
= \frac{1}{E_0 E} \left(\Lambda_0(E) -r_\textnormal{SN}\right) ;
\\
\left(\frac{\mathrm{d} \Lambda_1}{\mathrm{d} E}\right)^2
&= \frac{1}{E_0 E} \left( \sqrt{\frac{E}{E_0}}+1\right)^2
= \frac{1}{E_0 E} \left(\Lambda_1(E) -r_\textnormal{SN}\right) ;
\\
\left(\frac{\mathrm{d} \Lambda_0}{\mathrm{d} E}\right) \left(\frac{\mathrm{d} \Lambda_1}{\mathrm{d} E}\right)
&=\frac{1}{E_0 E} \left( \frac{E}{E_0} - 1 \right)
=\frac{1}{E_0 E}\left(\frac{\Lambda_0(E)+\Lambda_1(E)}{2}-2-r_\textnormal{SN}\right).
\end{aligned}
$$
Therefore, we can compute
$$
\begin{aligned}
&
\frac{\partial^2}{\partial E^2} f_s(\Lambda_0(E),\Lambda_1(E))
\\
&=
\frac{s(1-s)}{E_0 E}
\left[
 \Lambda_0^{s-2} \Lambda_1^{1-s}
 \left(\Lambda_0-r_\textnormal{SN}\right)
- 2 \Lambda_0^{s-1} \Lambda_1^{-s} \left(\frac{\Lambda_0+\Lambda_1}{2}-2-r_\textnormal{SN}\right)
+ \Lambda_0^s \Lambda_1^{-1-s} \left(\Lambda_1-r_\textnormal{SN}\right)
\right]
\\
&
\quad+
\frac{1}{2\sqrt{E_0 E^3}}
\left[
s\left(1-\Lambda_0^{s-1} \Lambda_1^{1-s}\right)+
(1-s)\left(1-\Lambda_0^s \Lambda_1^{-s}\right)
\right]
\\
&=
\frac{s(1-s)}{E_0 E}
\Lambda_0^{s} \Lambda_1^{1-s} \left[
4 \Lambda_0^{-1} \Lambda_1^{-1}
- r_\textnormal{SN}
\left(
\Lambda_0^{-2}
- 2 \Lambda_0^{-1} \Lambda_1^{-1}
+ \Lambda_1^{-2}
\right)
\right]
+
\frac{1}{2\sqrt{E_0 E^3}}
\left[
1 - \Lambda_0^s \Lambda_1^{1-s}
\left(
s \Lambda_0^{-1}
+(1-s) \Lambda_1^{-1}
\right)
\right]
\end{aligned}
$$
after canceling a few terms. Further collecting terms and simplifying yields
$$
\begin{aligned}
&
\frac{\partial^2}{\partial E^2} f_s(\Lambda_0(E),\Lambda_1(E))
\\
&=
\frac{ \Lambda_0^{s-1} \Lambda_1^{-s}  }{ 2 E_0 E \sqrt{E} }
\left[
2 \sqrt{E} s(1-s)
\left(
4
-
r_\textnormal{SN}
\Lambda_0 \Lambda_1
\left(
\Lambda_0^{-1}
- \Lambda_1^{-1}
\right)^2
\right)
+
\sqrt{E_0} 
\left( 
\Lambda_0^{1-s} \Lambda_1^s - s \Lambda_1 - (1-s) \Lambda_0
\right)
\right]
\\
&=
\frac{ \Lambda_0^{s-1} \Lambda_1^{-s}  }{ 2 \sqrt{E_0 E^3} }
\left[
8 \sqrt{\frac{E}{E_0}} s(1-s)
\left(
1
-
r_\textnormal{SN}
\frac{
\left(
\Lambda_1
- \Lambda_0
\right)^2
}{ 4 \Lambda_0 \Lambda_1}
\right)
-
\mathbb{C}_{1-s}
\left(
P_0^v \| P_1^v ; \mu_{\bbN_0}
\right)
\right]
\end{aligned}
$$
where we have used $ \mathbb{C}_{1-s}\left( P_0^v \| P_1^v ; \mu_{\bbN_0} \right) = (1-s) \Lambda_0 + s \Lambda_1 - \Lambda_0^{1-s} \Lambda_1^{s}  = \mathbb{C}_s\left( P_1^v \| P_0^v ; \mu_{\bbN_0} \right)$. Notice that
$$
\begin{aligned}
r_\textnormal{SN}
\frac{
\left(
\Lambda_1
- \Lambda_0
\right)^2
}{ 4 \Lambda_0 \Lambda_1}
&=
\frac{
\left( 4 \sqrt{E/E_0} \right)^2
r_\textnormal{SN}
}{
4
\left[\left(1-\sqrt{E/E_0}\right)^2 + r_\textnormal{SN} \right]
\left[\left(1+\sqrt{E/E_0}\right)^2
+ r_\textnormal{SN}\right]
}
\\
&=
\frac{
r_\textnormal{SN}
}{
\left(1-\sqrt{E/E_0}\right)^2 + r_\textnormal{SN}
}
\frac{
\left( 2 \sqrt{E/E_0} \right)^2
}{
\left(1+\sqrt{E/E_0}\right)^2 + r_\textnormal{SN}
}
< 1
\end{aligned}
$$ since $r_\textnormal{SN}>0$ and $E/E_0\in(0,1)$. Therefore, to show $\frac{\partial^2}{\partial E^2} f_s(\Lambda_0(E),\Lambda_1(E))> 0$ for $E\in(0,E_0)$, it suffices to show that for this range of $E$ we have
$$
\begin{aligned}
8 \sqrt{\frac{E}{E_0}} s(1-s)
\left(
1
-
r_\textnormal{SN}
\frac{
\left(
\Lambda_1
- \Lambda_0
\right)^2
}{ 4 \Lambda_0 \Lambda_1}
\right)
-
\mathbb{C}_{1-s}
\left(
P_0^v \| P_1^v ; \mu_{\bbN_0}
\right)    
>
0
.
\end{aligned}
$$
Rewriting the left-hand side as a function of $v = \sqrt{E/E_0}\in(0,1),$ we obtain
$$
\begin{aligned}
g_{s,r}(v) 
&\triangleq
8s(1-s)
v
\left(1-
\frac{r}{(1-v)^2+r}
\frac{4 v^2}{(1+v)^2+r}
\right)
\\&\quad
-
\left\lbrace
(1-s)[(1-v)^2+r]
+
s[(1+v)^2+r]
-
[(1-v)^2+r]^{1-s}
[(1+v)^2+r]^{s}
\right\rbrace
\end{aligned}
$$
where we have substituted $r \eqdef r_\textnormal{SN} >0$ for brevity.

Let us examine the regime where $r_\textnormal{SN} = r \ll 1$. Taking $r\to 0^+$ in $g_{s,r}(v)$, we get
$$
\begin{aligned}
    g_{s,0}(v)
    &\triangleq
    8s(1-s)v
    -
    [
    (1-s)(1-v)^2
    +
    s(1+v)^2
    -
    (1-v)^{2(1-s)}
    (1+v)^{2s}
    ]
    \\&=
    8 s v -8 s^2 v
    - 
    [
    ( 1 - v )^2 + 4 s v
    - \left( (1-v)^2 \right)^{1-s}
    \left( (1+v)^2 \right)^{s}
    ]
    \\&=
    4 s (1 - 2s) v
    + (1-v)^{2(1-s)}
    [
    (1+v)^{2s}
    -
    (1-v)^{2s}
    ]
    >
    0
\end{aligned}
$$
for $s\in(0,1/2]$ and $v\in(0,1)$. This implies that when the effect of dark count is minimal compared to that of the coherent state separation, and when $0<s\leq 1/2$, the Chernoff $s$-divergence $\mathbb{C}_s(P_0^v\|P_1^v;\mu_{\bbN_0})$ is \emph{strictly convex} w.r.t. the cost $E(v)$ for $v\in[0,1]$.

Since the mapping $w:[0,1]\to [0,1]:v\mapsto v^2$ is strictly increasing, by a change of variable $w=v^2$, the optimization problem 
\begin{align}
    \label{eq:opt_prob}
    \sup_{Q\in\mathcal{Q}(R_\textnormal{CA},R_\textnormal{CE})} \bbE_{V\sim Q} \left[\mathbb{C}_s(P_0^V\|P_1^V;\mu_{\bbN_0}) \right]
\end{align}
with $R_\textnormal{CA}=1$ and $R_\textnormal{CE}\leq 1$ can be recast as the equivalent optimization problem
\begin{align}
    \label{eq:equiv_opt_prob}
    \sup_{ Q' \in \mathcal{Q}'(R_\textnormal{CE})} \bbE_{W\sim Q'} \left[
    f_s\left(\Lambda_0(E_0 W), \Lambda_1(E_0 W)\right)
    \right]
\end{align}
where 
\begin{align*}
    \mathcal{Q}'(R_\textnormal{CE}) 
    &\eqdef
    \left\{
    Q': \mathscr{B}_{ [0,1] } \to [0,1]
    ~\bigg\vert~
    Q'( [0,1] )
    =
    1
    ,
    \bbE_{ W \sim Q' }
    \left[
    W
    \right]
    \leq 
    R_\textnormal{CE}
    \right\}
    .
\end{align*}
Consider the function $\tilde{f}_s : [0,1]\to \bbR_{\geq 0}: w\mapsto f_s\left(\Lambda_0(E_0 w), \Lambda_1(E_0 w)\right)$. Similar to the reasoning above, it is strictly convex; it is non-negative, being a divergence. Also, $\tilde{f}_s(0)=0$. Hence it is also strictly increasing on $w\in[0,1]$. Therefore, the maximizing distribution to \eqref{eq:equiv_opt_prob} is the time-sharing policy $Q^{'\star}=Q^{'\star}_s$ where
$$
\begin{aligned}
Q^{'\star}_s (w) = R_\textnormal{CE} \mathbb{I}[ w = 1 ] + \left( 1-R_\textnormal{CE} \right) \mathbb{I}[ w = 0 ]
,
\end{aligned}
$$
corresponding to the maximizing distribution to \eqref{eq:opt_prob} as the time-sharing policy $Q^{\star}=Q^{\star}_s$ where
$$
\begin{aligned}
Q^{\star}_s (v) = R_\textnormal{CE} \mathbb{I}[ v = 1 ] + \left( 1-R_\textnormal{CE} \right) \mathbb{I}[ v = 0 ]
.
\end{aligned}
$$
The proof for \textbf{Claim 1} is thus complete.

\textbf{Claim 2:} Fix any $Q\in\calQ(R_\textnormal{CA},R_\textnormal{CE})\setminus\{Q_0\}$ where $Q_0(v)\eqdef\bbI[v=0]$. Let 
\begin{align*}
    s^\star(Q) \eqdef \argmax_{s\in[0,1]} \bbE_{V\sim Q} \left[ \mathbb{C}_s(P_0^V\|P_1^V;\mu_{\bbN_0}) \right]
\end{align*}
Then $s^\star(Q) \in (0,1/2]$.

\textbf{Proof 2:} Recall the Chernoff $s$-divergence
\begin{align*}
    \mathbb{C}_s(P_0^v\|P_1^v;\mu_{\bbN_0})
    = s \Lambda_0 + (1-s) \Lambda_1
    - \Lambda_0^s \Lambda_1^{1-s}
    ,
\end{align*}
where $\Lambda_0=\Lambda_0(v)$ and $\Lambda_1=\Lambda_1(v)$. According to~\cite{nielsen_chernoff_2013,nielsen_information-geometric_2013}, this divergence is \emph{strictly concave} in $s\in[0,1]$ as long as $P_0^v\neq P_1^v \Leftrightarrow \Lambda_0(v)\neq \Lambda_1(v) \Leftrightarrow v\neq 0$. Therefore, for any fixed $v\neq 0$, the divergence has a unique maximum $s^\star(v) \eqdef \argmax_{s\in [0,1]} \mathbb{C}_s(P_0^v\|P_1^v;\mu_{\bbN_0})$, and at this maximum, the derivative of $\mathbb{C}_s(P_0^v\|P_1^v;\mu_{\bbN_0})$ w.r.t. $s$ is zero. Explicitly calculating this derivative gives
\begin{align*}
\frac{\partial}{\partial s}
\mathbb{C}_s(P_0^v\|P_1^v;\mu_{\bbN_0})
=
\Lambda_0 - \Lambda_1
- \Lambda_0^s \Lambda_1^{1-s} \log \left( \frac{ \Lambda_0 }{ \Lambda_1 }
\right)
.
\end{align*}
Evaluating at $s^\star(v)$ and set it to zero, we can solve for $s^\star(v)$ in terms of $\Lambda_0=\Lambda_0(v)$ and $\Lambda_1=\Lambda_1(v)$ as (cf. \cite{nielsen_chernoff_2013,nielsen_information-geometric_2013})
\begin{align*}
    s^\star(v)
    =
    \frac{
    \log \left(
    \frac{ ( \Lambda_0 / \Lambda_1 ) - 1 } { \log ( \Lambda_0 / \Lambda_1 ) }
    \right)
    }{
    \log ( \Lambda_0 / \Lambda_1 )
    }
    \triangleq
    S( R )
\end{align*}
where $R\eqdef R(v) \eqdef \Lambda_0 / \Lambda_1$ is the ratio between the Poisson rates, and $S:(0,1)\to [0,1]$ is the function
\begin{align*}
    S(R)
    \triangleq
    \frac{\log(\frac{R-1}{\log R})}{\log R}
    .
\end{align*}
Since $0 < \Lambda_0 < \Lambda_1$ for $v \in (0,1]$, indeed $R\in(0,1).$ It can shown through calculus that $S(R)$ is strictly increasing on $R\in(0,1)$, and $\lim_{R\to 0^+} S(R)=0$ while $\lim_{R\to 1^-} S(R)=1/2$. Hence, $S(R)\in(0,1/2)$ for $R\in(0,1)$, and we conclude that $s^\star(v)\in(0,1/2)$ for $v\neq 0$. Therefore, for any $Q_{\tilde{v}}(v)\eqdef\bbI[v=\tilde{v}]$ with $\tilde{v}\neq 0$, we have $s^\star(Q)=s^\star(\tilde{v})\in (0,1/2)\subset(0,1/2]$.

Next, we show that for any $Q$ that is not a point mass, the corresponding $s^\star(Q)$ is well-defined, and that $s^\star(Q)\in (0,1/2]$. Since $Q\neq Q_0$, the strict concavity of $\mathbb{C}_s(P_0^v\|P_1^v;\mu_{\bbN_0})$ on $s\in[0,1]$ when $v\neq 0$ carries over to $\bbE_{V\sim Q} \left[ \mathbb{C}_s(P_0^V\|P_1^V;\mu_{\bbN_0}) \right]$, and hence $\argmax_{s\in[0,1]} \left[ \mathbb{C}_s(P_0^V\|P_1^V;\mu_{\bbN_0}) \right]$ is a singleton, i.e., $s^\star(Q)$ well-defined.

Now we show that $s^\star(Q)\in(0,1/2]$. Suppose, towards a contradiction, that $s^\star(Q)\in(1/2,1).$ Fix any $v\neq 0$. Since $\mathbb{C}_s(P_0^v\|P_1^v;\mu_{\bbN_0})\geq 0$ is strictly concave on $s\in[0,1]$, and that its maximizer $s^\star(v)\in (0,1/2)$, it can be concluded that $s\mapsto\mathbb{C}_s(P_0^v\|P_1^v;\mu_{\bbN_0})$ is strictly decreasing on $[1/2,1]$, and hence $\mathbb{C}_{s^\star(Q)}(P_0^v\|P_1^v;\mu_{\bbN_0}) < \mathbb{C}_{1/2}(P_0^v\|P_1^v;\mu_{\bbN_0})$. Since this holds for every $v\neq 0$, and that $Q$ is not a point mass, we have
\begin{align*}
    \max_{s\in[0,1]}
    \bbE_{V\sim Q}
    \left[
    \mathbb{C}_{s}(P_0^V\|P_1^V;\mu_{\bbN_0})
    \right]
    &=
    \bbE_{V\sim Q}
    \left[
    \mathbb{C}_{s^\star(Q)}(P_0^V\|P_1^V;\mu_{\bbN_0})
    \right]
    \\&
    <
    \bbE_{V\sim Q}
    \left[
    \mathbb{C}_{1/2}(P_0^V\|P_1^V;\mu_{\bbN_0})
    \right]
    ,
\end{align*}
which is a contradiction. The proof for \textbf{Claim 2} is thus complete.

\textbf{Claim 1} and \textbf{Claim 2} together implies that, when $r=r_\textnormal{SN} \ll 1$ and hence $R_\textnormal{SN} \gg 1$ (high SNR regime), the joint optimization problem for the constrained open-loop exponent
\begin{align*}
    \beta_\textnormal{OL}
    &=
    \sup_{s\in[0,1], Q\in\calQ(R_\textnormal{CA},R_\textnormal{CE})}
        \bbE_{V\sim Q}
        \left[
        \bbC_{s} \left( P_{0}^{V} \| P_1^{V} ; \mu_{\bbN_0}
        \right)
        \right]
    ,
\end{align*}
which equals both iterated maximization problems in \eqref{eq:opt_prob_exponent} since the objective function is bounded (by uniform boundedness of $\bbC_{s} \left( P_{0}^{v} \| P_1^{v} ; \mu_{\bbN_0}
\right)$ shown in \eqref{eq:Cher_div_unif_bdd}), have joint maximizer(s) $(s^\star,Q^\star)$ satisfying $s^\star\in[0,1/2)$ and
\begin{align*}
Q^\star(v) = R_\textnormal{CE} \mathbb{I}[v=1] + \left(1-R_\textnormal{CE} \right) \mathbb{I}[v=0]
.
\end{align*}

\textbf{Claim 3:} 
There exists a pair $(r_\textnormal{SN},R_\textnormal{CE})$ such that the corresponding optimal policy is not a ``time-sharing" policy between $0$ and $\alpha$. 

\textbf{Proof 3:}
There exists a numerical example where $r_\textnormal{SN}=0.01$ and $R_\textnormal{CE}=0.9$ (still $R_\textnormal{CA}=1$) where the ``time-sharing" policy between $0$ and $\alpha$ is \emph{not} exponent-optimal. 

It can be readily computed that with $s^\star=0.31$ and $Q^\star(v)=\bbI[v=\sqrt{0.9}]$,
\begin{align*}
    \bbE_{V\sim Q^\star}
        \left[
        \bbC_{s^\star} \left( P_{0}^{V} \| P_1^{V} ; \mu_{\bbN_0}
        \right)
        \right]
    =
    \bbC_{s^\star} \left( P_{0}^{\sqrt{0.9}} \| P_1^{\sqrt{0.9}} ; \mu_{\bbN_0} \right)
    \approx
    1.9822
    ,
\end{align*}
while the time-sharing policy between $0$ and $\alpha$,
\begin{align*}
    Q_\textnormal{ts}(v)
    =
    0.9 \bbI[ v = 1 ] + 0.1 \bbI[ v=0 ]
\end{align*}
gives an exponent at most
\begin{align*}
    \max_{s\in[0,1]} 
    \bbE_{V\sim Q_\textnormal{ts}}
        \left[
        \bbC_{s} \left( P_{0}^{V} \| P_1^{V} ; \mu_{\bbN_0}
        \right)
        \right]
    &=
    0.9
    \max_{s\in[0,1]} 
        \bbC_{s} \left( P_{0}^{1} \| P_1^{1} ; \mu_{\bbN_0}
        \right)
    \\&=
    0.9 \times 2.1359
    \approx
    1.9314
    .
\end{align*}

\newpage
\bibliographystyle{IEEEtran}
\bibliography{references.bib}

\end{document}